%% file: example_paper.tex
\documentclass{article}

\usepackage{microtype}
\usepackage{graphicx}
\usepackage{subfigure}
\usepackage{booktabs} % for professional tables

\usepackage{hyperref}
\usepackage{verbatim}

\usepackage{array, makecell, multirow}
\newcolumntype{?}{!{\vrule width 1.2pt}}
\usepackage[table]{xcolor}

 \usepackage[accepted]{icml2023}

\usepackage{amsmath}
\usepackage{amssymb}
\usepackage{mathtools}
\usepackage{amsthm}

\usepackage[capitalize,noabbrev]{cleveref}

\theoremstyle{plain}
\newtheorem{theorem}{Theorem}[section]

\newtheorem{lemma}[theorem]{Lemma}

\theoremstyle{definition}
\newtheorem{definition}[theorem]{Definition}

\theoremstyle{remark}

\usepackage[textsize=tiny]{todonotes}

\icmltitlerunning{Fair $k$-Center: a Coreset Approach}

\usepackage{tablefootnote}

\newcommand{\Pot}{\mbox{Pot}}
\newcommand{\COL}{\mbox{\rm Col}}
\begin{document}

\twocolumn[
\icmltitle{Fair $k$-Center: a Coreset Approach in Low Dimensions}

% It is OKAY to include author information, even for blind
% submissions: the style file will automatically remove it for you
% unless you've provided the [accepted] option to the icml2023
% package.

% List of affiliations: The first argument should be a (short)
% identifier you will use later to specify author affiliations
% Academic affiliations should list Department, University, City, Region, Country
% Industry affiliations should list Company, City, Region, Country

% You can specify symbols, otherwise they are numbered in order.
% Ideally, you should not use this facility. Affiliations will be numbered
% in order of appearance and this is the preferred way.
\icmlsetsymbol{equal}{*}

\begin{icmlauthorlist}
\icmlauthor{Jinxiang Gan}{Hong Kong University of Science and Technology}
\icmlauthor{Mordecai Jay Golin}{Hong Kong University of Science and Technology}
\icmlauthor{Zonghan Yang}{Shanghai Jiao Tong University}
\icmlauthor{Yuhao Zhang}{Shanghai Jiao Tong University}

%\icmlauthor{}{sch}
%\icmlauthor{}{sch}
\end{icmlauthorlist}

\icmlaffiliation{Hong Kong University of Science and Technology}{Hong Kong University of Science and Technology, Hong Kong, China}

\icmlaffiliation{Shanghai Jiao Tong University}{Shanghai Jiao Tong University, Shanghai, China}

\icmlcorrespondingauthor{Jinxiang Gan}{jganad@connect.ust.hk}
%\icmlcorrespondingauthor{Firstname2 Lastname2}{first2.last2@www.uk}

% You may provide any keywords that you
% find helpful for describing your paper; these are used to populate
% the "keywords" metadata in the PDF but will not be shown in the document
\icmlkeywords{Machine Learning, ICML}

\vskip 0.3in
]

% this must go after the closing bracket ] following \twocolumn[ ...

% This command actually creates the footnote in the first column
% listing the affiliations and the copyright notice.
% The command takes one argument, which is text to display at the start of the footnote.
% The \icmlEqualContribution command is standard text for equal contribution.
% Remove it (just {}) if you do not need this facility.

\printAffiliationsAndNotice{}  % leave blank if no need to mention equal contribution
%\printAffiliationsAndNotice{\icmlEqualContribution} % otherwise use the standard text.

\begin{abstract}
%{\color{red}This document provides a basic paper template and submission guidelines.
%Abstracts must be a single paragraph, ideally between 4--6 sentences long.
%Gross violations will trigger corrections at the camera-ready phase.}

{Center-based clustering techniques are fundamental in some areas of machine learning such as data summarization.
Generic $k$-center algorithms can produce biased cluster representatives so there has been a recent interest in {\em fair} $k$-center clustering. Our main theoretical contributions are two new $(3+\epsilon)$-approximation algorithms for solving the fair $k$-center problem in (1) the dynamic incremental, i.e., one-pass streaming, model and (2) the MapReduce model. Our dynamic incremental algorithm is the first such algorithm for this problem (previous streaming algorithms required two passes) and our MapReduce one improves upon the previous approximation factor of $(17+\epsilon).$
Both algorithms work by maintaining a small {\em coreset} to represent the full point set and their analysis requires that the underlying metric has finite-doubling dimension. We also provide related heuristics for  higher dimensional data and experimental results that compare the performance  of our algorithms to existing ones.}

%When the underlying metric has a bounded doubling dimension, 
%our main contributions are $(3+\epsilon)$- approximation algorithms for the 
%our main contributions are: (1) a 2-round $(3+\epsilon)$ approximation algorithm for the MapReduce setting, which significantly improves the current $(17+\epsilon)$ theoretical bounds; and (2)
%(1) a dynamic incremental algorithm (that only permits seeing each data point once) that provides a  1-pass $(3+\epsilon)$ streaming algorithm, while the best-known streaming algorithm has to read the whole data set twice, i.e., 2-pass. To achieve a more straightforward implementation, we also develop practical heuristics based on our algorithms and they show advantages compared to previous algorithms even in some high-dimensional data sets.
\end{abstract}

    \input{introduction.tex}    

    \input{coreset}

    \input{mapreduce}
    
    \input{streaming}

 %   \input{slindingwindow}

    \input{experiment}

    \section{Future Direction}
    In this paper, we propose  a coreset-based algorithm framework for the fair $k$-center problem. By \cref{lem:approximation_in_coreset}, our approximation ratio for both the dynamic incremental and MapReduce algorithms will always be  essentially the same as that of the best static algorithm, which is currently $3.$ Any new improved static algorithm would therefore immediately  translate into an improvement to our algorithms. The current state of the art is that it is unknown whether $3$ is the best approximation that could be attained for the static fair $k$-center problem.  This needs to be further investigated.  In addition, our coreset techniques currently strongly require metrics with finite doubling dimensions.  Further work is needed to develop algorithms that do not have this requirement. Finally, our dynamic algorithm is only incremental and does not permit deletions. It would be useful to develop a fully dynamic fair $k$-center algorithm.

% In the unusual situation where you want a paper to appear in the
% references without citing it in the main text, use \nocite
\nocite{langley00}

\input{bib.tex}
%%%%%%%%%%%%%%%%%%%%%%%%%%%%%%%%%%%%%%%%%%%%%%%%%%%%%%%%%%%%%%%%%%%%%%%%%%%%%%%
%%%%%%%%%%%%%%%%%%%%%%%%%%%%%%%%%%%%%%%%%%%%%%%%%%%%%%%%%%%%%%%%%%%%%%%%%%%%%%%
% APPENDIX
%%%%%%%%%%%%%%%%%%%%%%%%%%%%%%%%%%%%%%%%%%%%%%%%%%%%%%%%%%%%%%%%%%%%%%%%%%%%%%%
%%%%%%%%%%%%%%%%%%%%%%%%%%%%%%%%%%%%%%%%%%%%%%%%%%%%%%%%%%%%%%%%%%%%%%%%%%%%%%%
\newpage
\appendix
\onecolumn
\input{appendix.tex}

%%%%%%%%%%%%%%%%%%%%%%%%%%%%%%%%%%%%%%%%%%%%%%%%%%%%%%%%%%%%%%%%%%%%%%%%%%%%%%%
%%%%%%%%%%%%%%%%%%%%%%%%%%%%%%%%%%%%%%%%%%%%%%%%%%%%%%%%%%%%%%%%%%%%%%%%%%%%%%%

\end{document}

%% file: introduction.tex
\section{Introduction}
    Data summarization is one of the most important problems in the area of machine learning. Its goal 
    %The goal of data summarization 
    is to compute a small set of data  which captures the key features of the original data set. %Performing further operations on a smaller data set is much easier and more efficient, e.g., we can directly run some machine learning algorithms on this small summary data set instead of reading the whole data set. 
    Performing further work, e.g., running  machine learning algorithms, on this small summary data set can be more efficient but almost as effective as  running them on the whole set.

    One issue with  standard data summarization algorithms is  that they often produce a summary which is {\em non-representative} of other aspects of the population  as a whole, e.g., they are biased with respect to attributes such as gender, race, and age (see, e.g., \cite{kay2015unequal}).
    and is therefore {\em unfair}. There has been much recent work in trying to alleviate this problem by developing technique for fair representation, in particular {\em fair $k$ center} (see e.g. \cite{kleindessner2019fair,chiplunkar2020solve,jones2020fair,angelidakis2022fair}).

    %In data summarization, naturally, how to select a set which is "approximately representative" of the whole data set is a central problem.  In current society, some results generated by some data sumarization approaches are biased with respect to attributes such as gender, race, and age (see, e.g., \cite{kay2015unequal}). Hence, the fairness in the data summarization draws a lot of attention. Specifically, the fair $k$ center problem becomes one of the most representative and effective methods in data summarization (see e.g. \cite{kleindessner2019fair,chiplunkar2020solve,jones2020fair,angelidakis2022fair}). 

    Going further, there is also interest in solving the fair $k$-center problem for {\em large} data sets, either using streaming algorithms (for one processor) or a large number of processors in parallel. That is the problem that we address in this paper.  In particular, we revisit the streaming and Map-reduce problems addressed in \cite{chiplunkar2020solve} and develop a new {\em coreset} based approach for metric spaces that have fixed doubling dimension  (Defined in  \cref{sec:notation}). This provides both better theoretical results and, in most practical examples, real performance.

    %In this paper, similar to \cite{chiplunkar2020solve}, we also study algorithms for the fair $k$ center problem in massive data such as the streaming algorithm, mapreduce algorithm and the sliding window algorithm. Compared with \cite{chiplunkar2020solve}, when the metric space of the data set has a finite doubling dimension ({\color{red}see definition in \cref{sec:notation}}), our algorithms show better performance and theoretical bounds.

    \subsection{Definition of the Fair $k$-Center Problem}
    %We first give the basic definition of the fair $k$ center problem.
    Let $(X,d)$ denote a  metric space and   $P\subset X$ be a set of points. Each point in $P$ belongs to exactly one of $m$ groups,  $\{1,...,m\}$. Let $g:X\rightarrow \{1,...,m\}$ denote the group assignment function. Each group $j$, has an associated  fixed capacity $k_j$ and  $k = \sum_{j=1}^m k_j$ . A (center) subset $S\subset P$ is called {\em feasible} if for every $j$, set $S$ contains at most $k_j$ points from group $j$. The goal is to compute a feasible set $S$ of centers so as to minimize $C(S)=\max_{p\in P} \min_{s\in S}d(p,s)$. 
    %$\max_{p\in P} \min_{s\in S}d(p,s)$ 
    $C(S)$ is called the cost of solution $S$.

    The special case, $m=1$, is  the well known and studied $k$-center problem.
    %which has been well studied. 
    
    Let $OPT$ denote the cost of an  optimal solution. An $\rho$-approximation algorithm for the problem  would find a feasible set of centers $C',$ such that $C(S') \le \rho\cdot OPT$. 

    In particular,  it is known that the plain  $k$-center problem ($m=1$)  is NP-hard to $(2-\epsilon)$-approximate \cite{hsu1979easy} for any $\epsilon >0,$  while there do exist some well known  2-approximation algorithms \cite{gonzalez1985clustering,hochbaum1985best} for solving it. %\mjgcomment{Is there known lower bound for fair $k$-center?}

    This paper studies the fair $k$-center problem in the MapReduce and streaming setttings.
    %Similar to \citet{chiplunkar2020solve}, in this paper, we revisit the fair $k$ center problem for large-scale datesets. Specifically, we also study two particularly popular models such as the MapReduce model and the streaming model. 
    The {\em MapReduce } model was  introduced by Google \cite{dean2008mapreduce}. In this setting, a set of processors process data in a sequence of parallel rounds on a large number of machines, each with only limited memory. In addition, only small amounts of inter-machine communication are permitted.
%lgorithm is not only highly parallel and distributed, but also requires only small amounts of inter-machine communication and only permits limited memory on each machine.
    The {\em streaming } model provides a mechanism to deal with large volumes of data in a limited-memory single-core processor by  restricting access to sequential passes over the data (with only  a limited amount of other working memory available). In particular, a {\em one-pass} streaming algorithm  may only see each piece of data once. One-pass streaming algorithms are essentially dynamic incremental algorithms that are only permitted limited working memory.
    % that we have already seen a significant portion of the stream and unlike the offline algorithms, we do not need to the whole stream when new points are inserted. A one-pass streaming algorithm is an incremental algorithm, 

    \cite{chiplunkar2020solve} study the fair $k$ center problem in those two models. They show a  $(3+\epsilon)$-approximation  two-pass streaming algorithm and a  $(17+\epsilon)$-approximation MapReduce algorithm. 
    %\mjgtext{as well as practical heuristics based on their alorithms}. {\marginpar{\textcolor{red}{Jinxiang: In their paper, they did not provide heuristics}}
    
    {In practice, it is known  that the metrics  in many real-world datasets possess  {\em finite doubling dimension}  (see definition in \cref{sec:notation}) \cite{talwar2004bypassing}. 
    %This constraint can be used to develop better algorithms.  
    Assuming finite doubling-dimension we develop better algorithms for the same problems. More specifically, %we show}
    
    \begin{itemize}
        \item we show a deterministic $(3+\epsilon)$-approximation {\em one-pass} streaming algorithm.  Unlike the best known
        $(3+\epsilon)$-approximation {\em two-pass} streaming algorithm of \cite{chiplunkar2020solve}
        this only accesses  each  data point  once and is actually  a  {\em dynamic incremental algorithm.}
        Although \cite{charikar2004incremental}  provides a dynamic incremental algorithms for the standard $k$-center problem, ours is the first such algorithm for the fair $k$-center one.
        
 %Comparing with the 2-pass streaming algorithm in \cite{chiplunkar2020solve}, a one-pass streaming algorithm is particularly interesting because it not only reads the big data set once less, but also supports the insertion of data. It is also called the {\em\color{blue} incremental} streaming algorithm.
        \item we show a deterministic  $(3+\epsilon)$-approximation MapReduce algorithm which  theoretically  {and practically} improves upon the $(17+\epsilon)$ approximation algorithm in \cite{chiplunkar2020solve}. Our MapReduce algorithm only has {one communication round.}
        %2 round algorithm.
        After each processor preprocesses its own internal data it sends a small summary to the coordinator. Combining the summaries of all processors,  the coordinator can generate a solution with a good  global approximation ratio.

        \item we run experiments to illustrate the practicality of our algorithms in both settings. More specifically, our theoretical guarantees only hold for fixed doubling-dimension, i.e., in low dimensions, so we also developed practical heuristics based on our algorithms  that work in higher dimensions {and  ran  experiments on them  using the same  data upon which  \cite{chiplunkar2020solve} was tested (including some  high-dimensional data sets) and provide a comparison.} 
        %The results are almost get at least as good as approximation while using less time.}
        
       % \mjgtext{I DECIDED IT WAS BETTER NOT TO DESCRIBE THE RESULTS OF YOUR EXPERIMENTS HERE.}
        
    \end{itemize}

    Our main tool is the coreset approach (also used by  \cite{ceccarello2019solving} to attack $k$-centers with outliers).
    %In both the streaming and MapReduce algorithm, we apply the {\color{blue}\bf coreset technique} (similar to \cite{ceccarello2019solving}).  

    {We conclude by noting that \cite{chiplunkar2020solve} proved that achieving a $(4-\epsilon)$-approximation to $k$-center in the MapReduce model with limited communication complexity is NP-hard.  The reason our $(3+\epsilon)$-approximation  does not violate their bound is that their proof assumed a general metric, while our algorithms assume metrics with bounded doubling dimension.}
    
   % are only for bounded constructed counterexample was in a general is NP-hard (unless     Note that in the MapReduce setting for the fair $k$ center problem, if communicate complexity is small enough, the approximation ratio is no better than 4 \cite{chiplunkar2020solve}.   Our algorithm computes a $(3+\epsilon)$ approximation solution because we use $O(m\ell k(\frac{8}{\epsilon})^D)$ space, where $\ell$ is the number of processors and $D$ is the doubling dimension (see definition in \cref{sec:notation}). 

    \subsection{Related Works}
    \cite{chen2016matroid} developed a  3-approximation algorithm that ran in  $O(n^2 \log n)$ time.
  \cite{kleindessner2019fair} then give a linear time algorithm with approximation ratio  $O(2^m)$, where $m$ is the number of groups in  the input. Finally, \cite{jones2020fair}  developed a faster, $O(nk)$ time, $3$-approximation algorithm. 
  {Note that 
  $3$ is still the best approximation factor known.}
  
   {Around the same time, \cite{chiplunkar2020solve} presented the previously discussed   $(3+\epsilon)$-approximation two-pass streaming algorithm and a $(17+\epsilon)$-Mapreduce algorithm.  \cite{yuan2021distributed} study the fair $k$ center problem {\em with outliers} and described  a $4$-approximation algorithm along with an  $18$-approximation distributed algorithm. Very recently, Angelidakis et al. \cite{angelidakis2022fair} combined the fairness constraint with a privacy constraint and proposed a new model called the {\em private and representative $k$-center} where the privacy constraint means that every selected center has to cover at least a given amount of data. They designed a $15$-approximation algorithm for this new model.}

    To conclude, we note that a  different fairness constraint is studied in \cite{chierichetti2017fair}, where the solution requires that proportion of groups in each cluster must be similar to that in the whole. Some other related works using this other fairness constraint can be found in \cite{bera2019fair,bercea2019cost,bera2022fair}.

    \section{Notation and Terminology}\label{sec:notation}
    %This paper studies a variant of the $k$ center problem which aims to find a partition.
 $P$ will always denote a finite point set in some underlying known metric space $(\mathcal{X},d).$  
    \begin{definition}
       $\mathcal{T}\subset 2^P$ is a partition of   $P$ if (1) $P=\bigcup_{S\in \mathcal{T}}S$; and (2) $\forall S_1,S_2\in \mathcal{T}$, $S_1\cap S_2=\emptyset$ 
    \end{definition}
    
    Our results assume that the underlying metric space $(\mathcal{X},d)$   has finite {\em doubling dimension}.
    %\mjgcomment{Do all of the examples in previous paper have finite doubling dimension?}
    %\paragraph{Doubling Dimensions} The doubling dimension of  metric space $(\mathcal{X},d)$ is the minimum value $\dim(\mathcal{X})$ such that any ball $B(x,r)$ in $(\mathcal{X},d)$ can be covered by $2^{\dim(\mathcal{X})}$ balls of radius $r/2$. 

    \begin{definition}
        [Doubling Dimensions] The doubling dimension of  metric space $(\mathcal{X},d)$ is the minimum value $\dim(\mathcal{X})$ such that any ball $B(x,r)$ in $(\mathcal{X},d)$ can be covered by $2^{\dim(\mathcal{X})}$ balls of radius $r/2$. 
    \end{definition}

    It is known that the doubling dimension of the Euclidean space $(R^D,\ell_2)$ is $\Theta(D)$ \cite{heinonen2001lectures}. 
    
    %{\color{blue}Given a $Y\subseteq \mathcal{X}$, we call the {\em aspect ratio} of the metric induced on $Y$ is $\frac{\max_{x,y\in Y}d(x,y)}{\min_{x,y\in Y}d(x,y)}$.}
    
    \begin{lemma}\cite{krauthgamer2004navigating}\label{lem:doubling}
        Let $(\mathcal{X},d)$ be a metric space and $Y\subseteq \mathcal{X}$.
        The {\em aspect ratio} of the metric induced on $Y$ is $\frac{\max_{x,y\in Y}d(x,y)}{\min_{x,y\in Y}d(x,y)}$.
        
        If the aspect ratio of $Y$ is at most $\Delta$ and $\Delta\geqslant 2$, then $|Y|\leqslant \Delta^{O(\dim(\mathcal{X}))}$.
        \end{lemma}

    In the sequel, $kCP$ and $FkCP$ respectively denote the $k$-center problem and fair-$k$ center problems. For  $P \subseteq \mathcal{X},$ 
  $r^*_{kC}(P)$ 
    and $r^*_{FkC}(P)$ 
    respectively denote the optimal values of $kCP$ and $FkCP.$
    Trivially, 
    $r^*_{kC}(P)\leqslant r^*_{FkC}(P)$.

%% file: coreset.tex
\section{Coreset Technique}
\label{sec:coreset}
        
    %A coreset $C$ is a subset of the input set $P$. 
    {The coreset paradigm is a well known and powerful tool for  studying large data sets by summarizing them using smaller ones.  For  $k$-centers,  a variant has previously been used to attack the $k$ center problem  with outliers \cite{ceccarello2019solving,ding2023randomized}). }
    %We first introduce the definition of the coreset.
    \begin{definition}[Coreset]
        For  $P\subset X$,  subset $C\subset P$ is an $\epsilon$-coreset of $P$ for $FkCP$, if for every  feasible set $S\subset P$ of points, 
        $$(1-\epsilon)\max_{p\in P}d(p,S)\leqslant \max_{p\in C}d(p,S)\leqslant (1+\epsilon)\max_{p\in P}d(p,S).$$
    \end{definition}
    
    $\epsilon$-coresets will be small subsets that approximate the original set. More specifically, we will see later, that solving the  $FkCP$ on an $\epsilon$-coreset of $P$  will, with some extra information,  yield an approximate solution for $P.$
    
    %Before we demonstrate our idea, we give a basic definition.
    We will first  need further definitions.
    \begin{definition}[$(r,\alpha)$-net] 
        Let $(\mathcal{X},d)$ be a metric space. For fixed parameter $r>0$, subset $Y\subseteq \mathcal{X}$ is an $(r,\alpha)$-net of $\mathcal{X}$ if it satisfies:\\
        $\bullet$ {(\em Packing Property:)} For every $x,y \in Y$, $d(x,y)\geqslant  r$;\\
        $\bullet$ {(\em  Covering Property:)} $\forall x\in \mathcal{X}$, there exists at least one $y\in Y$ such that $d(x,y)\leqslant \alpha\cdot r$. 
    \end{definition}
    When $\alpha=1$, this is the well known $r$-net from \cite{heinonen2001lectures}.

    In $FkCP$, the covering property will permit building an $\epsilon$-coreset  from an $(r,\alpha)$-net while the packing property restricts the number of points in the $(r,\alpha)$-net.
 
    \begin{lemma}\label{lem:net} Fix $P$ and let 
        $r'\leqslant r^*_{FkCP}(P)$. If
        $Y \subset P$ is an $(\frac{\epsilon}{\alpha} r',\alpha)$-net,  
        %and $Y\subseteq P$. Then, 
        then $Y$ is an $\epsilon$-coreset of $P.$ (see proof in appendix)
    \end{lemma}

    While $\epsilon$-coresets  as described do approximate $P$, they have lost all group information. To remedy this, we need the further definitions.

    %The significance of an $\epsilon$-coreset is that it compresses the original data set. When $Y\subseteq P$, once we solve $FkCP$ in an $\epsilon$-coreset, a good approximation solution for the original data set is constructed. However,  all points in an $(\epsilon r',\alpha)$-net $Y\subseteq P$ ($\forall r'\leqslant r^*$) could be from the same group. We cannot directly generate a feasible solution in $Y$. To fix this issue, we give a further definition.
    
    \begin{definition}\label{def:proper} 
        %Given a set $P$ of points, an $(r,\alpha)$-net $Y$ of $P$ is called $\epsilon$-proper if it satisfies the following  properties:
        Let $P$ be fixed and $Y \subset P$ be an  $(r,\alpha)$-net. $Y$  is called $\epsilon$-proper if 
$r\leqslant \frac{\epsilon}{\alpha} r^*_{FkC}(P)$.

        Further, for all $y \in Y$ associate  a 
        neighborhood set 
        %$N(y,r)=\{y\}$ 
        $N(y,r)$ 
        such that\\
        \ \quad $\bullet$ $y \in N(y,r)$\\
        \ \quad $\bullet$ If $p \in N(y,r)$, $d(p,y) \le \alpha r.$\\
        \ \quad $\bullet$ $\mathcal{T}=\{N(y,r): y\in Y\}$ is a partition of $P$

%        let $\mathcal{T}=\{N(y,r): y\in Y\}$ be  a partition of $P$ such that 
        
 %       \begin{itemize}
%        \item $Y\subseteq P$ and $r\leqslant \frac{\epsilon}{\alpha} r^*_{FkC}(P)$. 
 %       \item $\forall y\in Y$, we construct a neighborhood set $N(y,r)=\{y\}$. $\forall P\in P$, we arbitrarily select a point  $y\in Y$ such that $d(p,y)\leqslant \alpha r$ and let $N(y,r)=N(y,r)\cup \{p\}$. Noted that $\mathcal{T}=\{N(y,r): y\in Y\}$ is a partition of $P$.
        %$\bigcup_{x\in Y} N(x,r)=P$, and $N(x,r)\cap N(y,r)=\emptyset$ $\forall x,y\in Y$.
        %Each point $p\in P$ is either assigned to a point $x\in Y$ or $p\in Y$. Let $N(x,r)$ denote the set of all points in $P$ that are assigned to a point $x\in Y$. Noted that $x\in N(x,r)$, $\bigcup_{x\in C} N(x,r)=P$, and $N(x,r)\cap N(y,r)=\emptyset$ $\forall x,y\in C$.
%        \end{itemize}

Such a $\mathcal T$ always exists due to the covering property of $(r,\alpha)$-nets but might not be unique.  When discussing $\epsilon$-proper nets, we always assume an associated partition $\mathcal T.$
  \end{definition}  
   % If an $(r,\alpha)$-net of $P$ is $\epsilon$-proper, it is not only an $\epsilon$-coreset but also a partition of $P$. Currently, once a set $S\subset Y$ with $k$ points is not feasible, for each $x\in S$, $x$ can be replaced with a point $y\in N(x,r)$ from other groups. This replacement has an at most $(\epsilon/\alpha) r^*$ additive error. Thus, each point in an $\epsilon$-proper $(r,\alpha)$-net can represent a set of points from different groups. We give the following definition to describe this property.
    
    Note that  $C \subset Y$ might not be a coreset because it doesn't contain the correct number of points from each group.  In that case, if $Y$ is proper, we will be able to replace a point in $y \in C$ with a point in $N(y,r)$ that is close by.
    \begin{definition}\label{def:coreset_of_faircenter} Fix $P$ and let $Y$ be an 
       $\epsilon$-proper $(r,\alpha)$-net $Y$.
       %in the metric space $(X,d)$. 
       With every point $x \in Y$ associate %Associate with Each point in $Y$ associates with a function $\COL: C \rightarrow \Pi_{i=1}^m \{0,1\}$. $\forall x\in Y$, $\COL(x)=(x_1,x_2,...,x_m)$ is an $m$ dimensional vector satisfying:
       an $m$ dimensional vector %$\COL(x)$
       $\COL(x)=(\COL_1(x),\COL_2(x),...,\COL_m(x))$  defined by
        %$$ 
        %x_i=\left\{
        %    \begin{aligned}
         %   &1 &\quad &\text{there exists at least one point $p\in N(x,r)$ from group $i$}  \\
         %   &0   & &\text{otherwise}
         %   \end{aligned}
        %\right.
        %$$
        $$ 
        \COL_i(x)=\left\{
            \begin{aligned}
            &1 &\quad &  i \in \{g(p) \,:\, p \in  N(x,r)\}  \\
            &0   & &\text{otherwise}
            \end{aligned}
        \right.
        $$
      %Let $\COL_i(x)=x_i$.
       Furthermore, when $\COL_i(x)=1$, a point $y\in N(x,r)$ from group $i$ will be stored in
       $\Pot(x,i)$ as follows: 
        $$ 
        \Pot(x,i)=\left\{
            \begin{aligned}
            &x  & &i=g(x)\\
            &\mbox{undefined} &   & \COL_i(x)=0  \\
            & \text{any point $p\in N(x,r)$  with $g(p)=i$} &  & \text{Otherwise}
            \end{aligned}
        \right.
        $$
     Finally, define 
     %$\Pot(x)=\cup_{i}\{\Pot(x,i)\}$, 
     $\Pot(x)=\{\Pot(x,i)\,:\, \COL_i(x)=1\}.$
     %where the union is over all $i$ for which $P(x,i)$ is defined.
    \end{definition}

    %For a subset $X\subset Y$, we can also define $$\COL(X)=(\bigvee_{x\in X}\COL_1(x),...,\bigvee_{x\in X}\COL_m(x)).$$ 
   % After we define $\COL(\cdot)$ on each point in an $\epsilon$-proper $(r,\alpha)$-net $Y$, we can solve $FkCP$ using $Y$, and the solution is a good approximation solution for the original data set. We now give a definition of $FkCP$ on a $\epsilon$-proper $(r,\alpha)$-net.
    We require one further definition
    \begin{definition}\label{def:candidate}
    Fix $P.$ Let $Y \subset P$ be $\epsilon$-proper.
        %Given a $\epsilon$-proper $(r,\alpha)$-net, each $y\in Y$ has an associated group vector $\COL(y)$. 
        We say  $y$ is from group $i$ if $\COL_i(y)=1$ ($\forall 1\leqslant i\leqslant m$). 
        
         $S\subset Y$ is a { \em candidate feasible solution} of $Y$ if there exists a feasible set $S'\subseteq P$ such that \\
          $\bullet$ $S'\subseteq \bigcup_{s\in S}N(s,r)$\\
         $\bullet$ $|S'\cap N(s,r)|=1 \quad \forall s\in S$\\
        Note that $|S'|=|S|$. We define the cost of the candidate feasible solution  for $Y$ is $\max_{y\in Y}d(y,S)$.
%        \mjgtext{Why do we need the second condition and the fact that  $|S'|=|S|$? Isn't the fact that $S'$ is feasible enough?}
    \end{definition}
    
   \begin{lemma}  \label{lem:CandHelper}
    Fix $P$ and let $Y$ be an $\epsilon$-proper $(r,\alpha)$-net.
      %Suppose that we have an $\epsilon$-proper $r$-net $Y$ and each $y\in Y$ has associated $\COL(y)$ and $N(y,r)$.
      
      %If there exists an algorithm that can generate a candidate feasible solution $S$ with at most $\phi\cdot r^*_{FkC}(P)$ cost, 
      { If there exists a candidate feasible solution $S$ with cost $c$  and associated feasible $S'$ as defined in \cref{def:candidate}
      %we can construct a feasible solution of $FkCP$ in $P$ with
      then $S'$ is a feasible solution in $P$ with 
      %cost at most 
      $ C(S') \le c + 2\epsilon r^*_{FkC}(P)$}
      %at most $(\phi+2\epsilon) r^*_{FkC}(P)$ cost.
      (see proof in appendix)
   \end{lemma}

   We will now show that  if we can solve $FkCP$ on  (a variant of) an  $\epsilon$-proper $(r,\alpha)$-net $Y$ of $P,$ something which will be very small,  we can easily  get a good approximate solution for $FkCP$ on the original data set $P$. 
   
   %{\color{blue} 
   %Currently, the running time highly depends on the size of the coreset instead of the original large scale data set}
  % \mjgtext{Need to mention run time here, or this doesn't really make sense}

    \begin{lemma}\label{lem:approximation_in_coreset}
    
    Let $A$ be a $\rho$-approximation algorithm for $FkCP$ and 
    $T_A(n)$ its running time on an input of size $n.$ Then, given an  
    $\epsilon$-proper $(r,\alpha)$-net $Y$  
    for $P,$ we can create a 
    $\rho(1+3\epsilon)r^*_{FkC}(P))$-approximation 
    algorithm for solving $FkCP(P)$
   in time $T_A(m|Y|)+O(m|Y|).$
      %  Given an $\rho$-approximation algorithm for $FkCP$, we can compute a \mjgtext{candidate feasible solution $S$  with cost at most $\rho(1+\epsilon)r^*_{FkC}(P)$and and its associated feasible set $S'$  using only an additional $O(m|Y|)$ time}.
    \end{lemma}
    
    \begin{proof}% Starting with $\epsilon$-proper $(r,\alpha)$-net $Y$, 
        Create a  new set  $Y'$  as follows.  For each point $y \in Y$  and each color $i$ such that $\COL_i(y)=1$ add a new point $y'$ to $y.$ $y'$ will be at the same location as $y$ and be in group $i.$  We say that $y$ is {\em associated} with $y'$.
        %in which all the points in $Y'$ are at the same location aswhere there are $|Y|$ locations and each location have at most $m$ points. For each $y\in Y$, there exists a location $y'$ in $Y'$ located at the place of $y$. $\forall 1\leqslant i\leqslant m$, if $COLOR_i(y)=1$, we construct a point $y_i'$ from group $i$ at the location $y'$. 
        Note that   $ |Y'| = O(m|Y|)$. 
        %and we can run algorithm called $A$ in $Y'$.

        %the problem in $P$. 
        
        Let $O=\{o_1,...,o_k\}$ denote the optimal solution of $FkCP(P).$ By the definition of $Y,$
       $o_t \in N(y(o_t),r)$ for some $y(o_t) \in Y.$
       %Since $O$ is feasible,  $S = \{y(o_t)\,:\,o_t \in O\}$ is a candidate feasible solution.

       By the definition of $Y'$ there exists $y'(o_t) \in Y'$ (located at $y(o_t)$) such that 
        $g(y'(o_t)) =g(o_t)$ and  $d(y'(o_t),o_t)\leqslant \alpha r$.
        
         Now feed $Y'$ as input to the $\rho$-approximate  $FkCP$ algorithm. Call this algorithm $A.$ Let $A(Y')$ denote the value of the solution computed by algorithm $A$ for input $Y'$.  $A(Y')\leqslant \rho r^*_{FkC}(Y')$.       
        %, with $g(o_t)=j.$ Then $o_t \in N(y(o_t),r)$ for some $y(o_t) \in Y$ so there exists $y'(o_t) \in Y'$ with 
        %$g(y'(o_t)) =j$ and  $d(y'(o_j),o_j)\leqslant \alpha r$.
        %there exists $y'(o_t)\in Y'$ such that 
        %$\forall o_j\in O$ from group $j$, there exists $y'(o_j)\in Y'$ from group $i$ such that $d(y'(o_j),o_j)\leqslant \alpha r$. 
     Because $O$ is feasible,  $O'=\{y'(o_1),y'(o_2),...,y'(o_k)\} \subset Y'$ is feasible and $\forall y'\in Y'$ $d(y',O')\leqslant d(y',O)+\alpha r\leqslant r^*_{FkC}(P) +\alpha r$.
        Thus, 
        $$\begin{aligned}
            A(Y')&\leqslant \rho r^*_{FkC}(Y') \leqslant \rho \max_{y'\in Y'}d(y',S') \\
            &\leqslant \rho (r^*_{FkC}(P) +\alpha r)\leqslant \rho (1+\epsilon)r^*_{FkC}(P).
        \end{aligned}$$
        Finally, let $\bar S$ be the actual feasible solution generated by algorithm $A$ run on $Y'$ and $S \subset Y$ the set of points associated with the points in $\bar S.$ For each $y \in S$, arbitrarily choose one point $y'$ from $\bar S$ associated with $y$ and add $\Pot(x,g(y'))$ to $S'$.  Since $\bar S$ is feasible (in $Y'$), $S'$ is feasible (in $P$).
        This $S'$ witnesses that  $S$ is a candidate feasible solution of $Y$.
Furthermore, since
        $$A(Y') = \max_{y' \in Y'} d(y',S')= \max_{y \in Y} d(y,S),$$
the cost of $S$ for $y$ is $\le \rho (1+\epsilon)r^*_{FkC}(P).$  

    {Plugging this $S,S'$ into \cref{lem:CandHelper}  completes the construction. Note that all of the work performed other than calling $A(Y')$ can be implemented in $O(m |Y|)$ time.}
    % \mjgtext{The original proof had a lot of issues and I totally rewrote it. Please check.}
    \end{proof}

    {Combining the last lemma with  the $O(kn)$-time $3$-approximation JNN algorithm from  \cite{jones2020fair} will yield good approximate solutions for  $FkCP(P)$  given an $\epsilon$-proper $(r,\alpha)$-net of $P$. It remains to construct such nets.}
    
   %The remaining question is to construct an $\epsilon$-proper $(r,\alpha)$-net of $P$. The construction of an $(r,\alpha)$-net depends on $r$.
   When $r$ is fixed, it is easy to construct an $(r,\alpha)$-net $Y$ from scratch. There are many scenarios, though, where it is more desirable  to build the nets by {\em merging} previously built ones. This occurs in both the MapReduce and streaming models.
   
   %occurs, for example, in the MapReduce  model, in which we want to merge nets from each processor to construct a  net for the entire set.  It also occurs in the incremental model in which we want to build a new net that results  after adding a new point.
%o read each point at most once. With more points are read, $r$ that we set would be changed. When $r$ is changed to a larger one $r'$, 
   The following algorithm/lemma  will be  a useful tool when constructing new nets from old ones.
   %shows that we can directly construct an $(r',\alpha)$-net from the current $(r,\alpha)$-net rather than revisit all points again.

    %\mjgtext{I removed the old Lemma 3.9 and put its content in the next lemma. Please check that the  old Lemma 3.9 was not used anywhere else.}
    %\begin{lemma}\label{lem:property_of_nets}
    %Let $ R\geqslant 2r$ and 
    %%Given a metric space $(X,d)$ and a point set $P\subset X$, suppose that 
    %$Y_1$ be an $(r,2\alpha)$-net of $P$ and $Y_2$ be  an $(R, \alpha)$-net of $Y_1$.
    %%where $ R\geqslant 2r$. 
    %Then, $Y_2$ is an  $(R,2 \alpha)$-net of $P$.
    %\end{lemma}
    %\begin{proof} Fix $x \in P.$ Because $Y_1$ is an $(r,2\alpha)$-net of $P$,
        %there exists  $y_1\in Y_1$ such that $d(x,y_1)\leqslant 2\alpha r$, 
       % Because $Y_2$ is a $(R, \alpha)$-net of $Y_1$, there exists  $y_2\in Y_2$ such that $d(y_1,y_2)\leqslant \alpha R$. Thus, %$$d(x,Y_2)\leqslant d(x,y_2)\leqslant d(x,y_1)+d(y_1,y_2)\leqslant 2\alpha r+\alpha R\leqslant 2\alpha R$$
       % In addition,  because  $Y_2$ is a $(R, \alpha)$-net of $Y_1$, $\forall x,y\in Y_2$, $d(x,y)\geqslant R$ . Thus , $Y_2$ is an %$(R, 2\alpha)$-net of $P$.
   % \end{proof}
    
   %% In \cref{algo:coreset}, when $Y_1=P$, $Y_2=\emptyset$ and {\color{blue}$R\geqslant 2r$}, we construct an $\epsilon$-proper $(R, 2\alpha)$-net of $P$ from an $\epsilon$-proper $(r, 2\alpha)$-net of $P$. Additionally, if both of nets in the input are $\epsilon$-proper, the output is also  $\epsilon$-proper,
    
    \begin{algorithm}[h!]
	\caption{Construct $\epsilon$-proper $(R, 2\alpha)$-net $Y'$ of $P_1 \cup P_2$}
	\label{algo:coreset}
    {\bf Input:} An $\epsilon$-proper $(r,2\alpha)$-net $Y_1$ of  $P_1$ and an  $\epsilon$-proper $(R,2\alpha)$-net $Y_2$ of  $P_2$. %Each point $y\in Y_1\cup Y_2$ has an associated vector $\COL(y)$. 
	\begin{algorithmic}[1]
	\STATE Set $Y'=Y_2$
    \FOR{each $y\in Y_1$}
    \IF{there exists $y'\in Y'$ such that $d(y,y')\leqslant \alpha R$}
%    \STATE {\mjgtext{$N(y',R)=N(y',R) \cup N(y,r)$}}
    \FOR{$1\leqslant i\leqslant m$}
    \IF{$\COL_{i}(y')=0$ and $\COL_{i}(y)=1$}
    \STATE $\COL_i(y')=1$ and  $\Pot(y',i)=\Pot(y,i)$
    \ENDIF
    \ENDFOR
    \ELSE
    \STATE $Y'=Y'\cup \{y\}$
    \ENDIF
    \ENDFOR
	\end{algorithmic}
    \end{algorithm}
    
    \begin{lemma}\label{lem:update_r}
         %$r^*$ is the optimal value of the fair $k$ center in $P$. 
         Let $Y_1$ be an  $\epsilon$-proper $(r,2\alpha)$-net of  $P_1$ and $Y_2$ an  $\epsilon$-proper $(R,2\alpha)$-net of  $P_2$.
         If $2r\leqslant R\leqslant \frac{\epsilon}{2\alpha}r^*_{FkC}(P)$ and $\alpha\geqslant 1$, $Y'$ constructed by \cref{algo:coreset} is an $\epsilon$-proper $(R, 2\alpha)$-net of $P_1\cup P_2$ whose $\COL$ and $\Pot$ vectors % and $N(y',R)$ sets 
         are accurately updated.
    \end{lemma}
    \begin{proof}
    Let $Y_1'=Y' \cap Y_1$ be the points from $Y_1$ added to $Y'.$ Now let $y,y' \in Y'.$ 
    If $y,y' \in Y_2$ then $d(y,y') \ge R.$  
    If $y\in Y_2$  and $y'\in Y_1'$ then by construction, $d(y,y')> \alpha R > R.$  If both $y,y' \in Y_1'$  assume that $y$ was added to $Y'$ before $y'$. Then, again, by construction, $d(y,y')> \alpha R > R.$ 
    So, in all cases, the packing condition $d(y,y') \ge R$ holds.

    To validate the covering  condition, first assume that $y \in P_2.$
    Then, because $Y_2$ is an  $\epsilon$-proper $(R,2\alpha)$-net of  $P_2$, there exists $y' \in Y_2 \subseteq Y'$ 
    such that $d(y,y') \le 2 \alpha R.$ 
    
    Next assume that $y \in P_1.$ Because $Y_1$ is an  $\epsilon$-proper $(r,2\alpha)$-net of  $P_1$, there exists $y' \in Y_1$ 
    such that $d(y,y') \le 2\alpha r.$  If $ y' \in Y_1'$ then, since $2\alpha r \le 2 \alpha R$, the covering  condition trivially  holds.
    If $ y' \not\in Y_1'$, then there exists $\bar y \in Y'$ such that $d(\bar y,y') \le \alpha R.$ But then,
    $$d(y,\bar y) \le d(y,y') + d(y',\bar y)\le 2 \alpha r + \alpha R \le 2 \alpha R.$$
    Thus the covering condition always holds and $Y'$ is an $(R, 2\alpha)$-net of $P_1\cup P_2$. It is proper  because $R\leqslant \frac{\epsilon}{2\alpha}r^*$.

        %Each point $y\in Y_1$ read on line 2 is either close to a point in $Y'$ or added into $Y'$. Thus at the conclusion, $\forall y\in Y_1$, there exists $y'\in Y'$ such that $d(y,y')\leqslant \alpha R$. Additionally, because $\alpha\geqslant 1$, $\forall{y_1,y_2}\in Y'$ $d(y_1,y_2)\geqslant \alpha R\geqslant R$.   Thus, $Y'$ is an $(R,\alpha)$ net of $ Y_1$.  Furthermore, since $R\geqslant 2r$, by lemma \ref{lem:property_of_nets}, $Y'$ is an $(R,2\alpha)$-net of $P_1$. We also know that $Y_2\subseteq Y'$ is the $(R,2\alpha)$-net of $P_2$, so $Y'$ is an $(R,2\alpha)$-net of $P_1\cup P_2$.
        
        That the $\COL$ and $\Pot$ vectors % and $N(y',R)$ sets 
        are accurately updated for $Y'$ follows directly 
        from the definitions and the fact that, if $y \in Y_1$ is not added to 
        $Y'$ because $d(y,y') \le \alpha R$ for some $y' \in Y',$ then all points from $P_1$ in $N(y,r)$ are within distance $2\alpha R$ of $y'.$
        
        % $\forall y'\in Y'$, initially define $D(y')=\{y'\}$. $\forall y\in Y_1/Y'$, let $y'\in Y'$ denote the point selected by line 3 of \cref{algo:coreset} and  add $y$ into $D(y')$. Then, we can define $N(y',R)=\bigcup_{y\in D(y')}N(y,r)$. Since $\bigcup_{y\in Y_1}N(y,r)=P_1$ and $\forall y_1,y_2\in Y\quad N(y_1,r)\cap N(y_2,r)=\emptyset$, we have $\bigcup_{y'\in Y'}N(y',r)=P_1\cup P_2$ and $\forall y'_1.y'_2\in Y' N(y'_1,r)\cap N(y'_2,r)=\emptyset$.
        
       % In addition, $\forall y'\in Y'$, if there exists a point $p\in N(y',R)$ from group $i$, there exists $y\in D(y')$ such that $p\in N(y,r)$. Hence, $\COL_i(y')=\COL_i(y)=1$ and $\Pot(y')\neq \emptyset$. If there does not exist a point $p\in N(y',R)$ from group $i$, $\forall y\in D(y')$ $N(y,r)$ does not contain a point from group $i$, i.e., $\COL_i(y')=\COL_i(y)=0$.
        
        %Finally, because $R\leqslant \frac{\epsilon}{2\alpha}r^*$, $Y'$ is an $\epsilon$-proper $(R, 2\alpha)$-net of $P_1\cup P_2$.
    \end{proof}

   % \mjgtext{I totally rewrote this proof. Please check it}
    
    By \cref{algo:coreset} and \cref{lem:update_r}, when $r$ is updated we can efficiently construct a new $\epsilon$-proper $(r, \alpha)$-net of $P$ from $Y_1,Y_2$ in  time $O(m|Y_1|\cdot|Y_1\cup Y_2|)$. %\mjgtext{Please fix this running time!}
    In the next sections, we describe how to use these tools to  construct an $\epsilon$-proper $(r, \alpha)$-net of $P$ in streaming and MapReduce settings.

%% file: mapreduce.tex
\section{The  MapReduce Setting}
    In the MapReduce model of computation, the set $P$ of points to be clustered is distributed equally among $\ell$ {\em processors}. Each processor is allowed restricted access to the metric $d$: it may only compute the distance between only its own points. Each processor performs some computation on its set of points and sends a summary of small size to a {\em coordinator}. From the summaries, the coordinator then computes a globally feasible set $S$ of points which covers all the $n$ points in $P$ within a small radius. Let $P_t$ denote the set of points distributed to processor $t$.
    
    \subsection{Robust Setting}
    Firstly, given any constant $\epsilon>0$, we present a $3(1+\epsilon)$-approximation algorithm in the MapReduce setting. In this subsection, robustly set a target ratio $3(1+\epsilon)$ in advance and define $\Bar{\epsilon}=\epsilon/3$. The algorithm constructs a coreset with size $O(k\ell (8/\Bar{\epsilon})^D)$ where $D$ is the doubling dimension of the metric space and $\ell$ is the number of processors in the MapReduce setting.
    
	\begin{algorithm}[h!]
	\caption{Computation by the $t$'th Processor} 
	\label{ALg:Map}
	%\hspace*{0.02in} 
	{\bf Input:}  Set $P_i$, metric $d$ restricted to $P_i$, group assignment function $g$ restricted to $P_t$
	%{\bf Output:} A $(1+\epsilon)$-approximate minimum enclosing ball $B(c,r)$ containing all points in $P$.
	\begin{algorithmic}[1]
    \STATE Arbitrarily select a point $p_1^t$ from $P_t$ and set $S_t=Y_t=\{p_1^t\}$
    \FOR{$j=2$ to $k$}
    \STATE Compute $p_j^i \leftarrow \arg\max_{p\in P_t} d(p,S_t)$;
    \STATE Set $S_t=S_t\cup \{p_j^t\}$
    \ENDFOR
    \STATE Compute $r_t=\frac{1}{8}\max_{p\in P_t}d(p,S_t)$
    \STATE Set $\COL_{g(p_1^t)}(p_1^t)=1$ and $\Pot(p_1^t,g(p_1^t))=p_1^t$
    \STATE Set  $\COL_i(p_1^t)=0$ $(\forall i\neq g(p_1^t))$
    \FOR{each $p\in P_t$}
    \IF{there exists $y\in Y_t$ such that $d(p,y)\leqslant 2\Bar{\epsilon} r_t$}
 %   \STATE \mjgtext{$N(y,\epsilon r_t)=N(y,\epsilon r_t)\cup \{p\}$}
    \IF{$\COL_{g(p)}(y)=0$}
    \STATE $\COL_{g(p)}(y)=1$ and $\Pot(y,g(p))=p$
    \ENDIF
    \ELSE
    \STATE $Y_t=Y_t\cup \{p\};$ %\mjgtext{$N(p,\epsilon r_t) = \{p\}$}
    \STATE Set $\COL_{g(p)}(p)=1$ and $\Pot(p,g(p))=p$
    \STATE Set $\COL_i(p)=0$ $(\forall i\neq g(p))$
    \ENDIF
    \ENDFOR
    \STATE Send $(Y_t,r_t)$ to the coordinator, where each $y\in Y_t$ associates with a vector $\COL(y)$ and a set $\Pot(y)$.
	\end{algorithmic}
    \end{algorithm}	
    
    \begin{lemma} \label{lem:coreset_of_MapReduce}
        \Cref{ALg:Map} computes an $\Bar{\epsilon}$-proper $(\Bar{\epsilon} r_t,2)$-net $Y_i$ of given $P_t$, where $|Y_t|=O(k(8/\Bar{\epsilon})^D)$. (see proof in appendix)
    \end{lemma}

     Since each point $y_t\in Y_t$ has an associated set $\Pot(y_t)$, by \cref{lem:coreset_of_MapReduce} processor $t$ sends $O(mk(8/\Bar{\epsilon})^D)$ points to the coordinator. After receiving information from all processors, the coordinator will use \cref{lem:update_r}
     to compute an $\Bar{\epsilon}$-proper %$(\epsilon r,2)$-
     net $Y$ of the input set $P$ and solve $FkPC$ in this coreset. 
     %\mjgtext
     {To use the lemma, we first need to  lower bound $r^*_{kC}(P)$.}
     %{\color{blue} To achieve an appropriate net of $P$, we need to estimate the lower bound of $r^*_{kC}(P)$. The following shows that the greedy $kPC$ algorithm from \cite{gonzalez1985clustering} supports us in evaluating this value.

%\mjgtext{In proof you wrote $\forall q\in Q.$ I changed this to $\exists q \in Q$. Please check.}
     \begin{lemma}\label{lem:property_of_greedy_algo}
         {$\forall Q\subset P$, let $S$ and $A(Q)$ respectively denote the solution set and the value returned by the $2$-approximation greedy $kPC$ algorithm \cite{gonzalez1985clustering} when running on $Q$ (recall that this is lines 2-5 of \cref{ALg:Map}). Then $A(Q)\leqslant 2\cdot r^*_{kPC}(P)$. } (see proof in appendix)
     \end{lemma}

	\begin{algorithm}[h!]
	\caption{Computation by the coordinator} 
	\label{ALg:Coordinator}
	%\hspace*{0.02in} 
	{\bf Input:}  $\forall 1\leqslant t\leqslant \ell$, an $\Bar{\epsilon}$-proper $(\Bar{\epsilon} r_t,2)$ net $Y_t$ of $P_t$ and each $y\in Y_t$ has associated $\COL(y)$ and $\Pot(y)$
	%{\bf Output:} A $(1+\epsilon)$-approximate minimum enclosing ball $B(c,r)$ containing all points in $P$.
	\begin{algorithmic}[1]
    \STATE Set $Y=\emptyset$ and $M=\emptyset$
    \STATE Let $R=2\cdot\max_{1\leqslant t\leqslant \ell}r_t$
    \FOR{$1\leqslant t\leqslant \ell$}
    \STATE Apply \cref{algo:coreset} in $Y_t$ and $Y$ to construct a new $\Bar{\epsilon}$-proper $(\Bar{\epsilon} R,2)$ net $Y$ of $M\cup P_t$.
    \ENDFOR
	\end{algorithmic}
    \end{algorithm}

    \begin{lemma}\label{lem:coreset_of_coordinator}
        \cref{ALg:Coordinator} returns an $\Bar{\epsilon}$-proper $(\Bar{\epsilon} R,2)$ net $Y$ of $P$ in time $O(m\ell k^2 (8/\Bar{\epsilon})^{2D})$. (see proof in appendix)%, where {\color{blue}$|Y|=O(k\ell (4/\epsilon)^D)$. }
    \end{lemma}
    %\mjgtext{Should this be $|Y|=O(k\ell (4/\epsilon)^D)$?}

    %\mjgtext{Why.  You also need to show that $ \epsilon R\leqslant \frac{\epsilon}{2\alpha}r^*_{FkC}(P)$, or, equivalently, that 
    % $ R\leqslant \frac{1}{2}r^*_{FkC}(P)$, 
    %which is missing. }

    After each processor runs \cref{ALg:Map} and the coordinator runs 
    \cref{ALg:Coordinator} the coordinator then  
    uses \cref{lem:approximation_in_coreset} with  the $3$-approximation JNN   algorithm \cite{jones2020fair} for the fair $k$-center problem. When $\Bar{\epsilon}=\epsilon/3$, this immediately returns a 
    $3(1+\epsilon)$-approximate solution to the fair $k$-center problem on $P.$ 
    Recall that  the running time of the JNN algorithm is $O(|X|k)$ where $|X|$ is the number of points in the input set. The coordinator receives $O(\ell k(24/\epsilon)^D)$ points and the $Y$ outputted by \cref{ALg:Coordinator} is a subset of these. Hence, the use of \cref{lem:approximation_in_coreset} requires only  $O(m\ell k^2(24/\epsilon)^D)$ time.

    %\mjgtext{I JUST REWROTE THE PARAGRAPH ABOVE. PLEASE CHECK}
    %\mjgtext{Questions:
    %(a) \cref{lem:approximation_in_coreset} gives a $3(1+ 3 \epsilon)$-approximation and not a $3(1+  \epsilon)$ one.  (I changed that a few days ago to be consistent with the proof) This needs to be fixed. If it is $3(1+ 3 \epsilon)$ then changing that to $3(1+ \epsilon)$ would require changing the space by replacing $\epsilon$ with $3 \epsilon$ in the space bound.\\
    %(b) The running time here needs to be fixed.  There are a few issues.
    %The first is that I don't think you actually need the $\ell$ in the $O(m\ell k^2(8/\epsilon)^D)$ space bound. $Y$ is $\epsilon$-proper, so it shoudn't have $\ell$ involved.\\
    %(c)  You should mention something about the running time of \cref{ALg:Coordinator}.
    %}
    
    %{\color{blue} Therefore, we can compute the fair $k$ center problem by running JNN algorithm \cite{jones2020fair}, which is a 3 approximation algorithm for the the fair $k$ center problem, on the coreset $Y$ constructed by \cref{ALg:Coordinator}. By lemma \ref{lem:approximation_in_coreset}, JNN algorithm returns a $3(1+\epsilon)$ approximation solution. Additionally, the running time of JNN algorithm is $O(|X|k)$ where $|X|$ is the number of points in the input set. Note that the coordinator receives $O(m\ell k(8/\epsilon)^D)$ points and $Y$ is its subset. Hence, JNN algorithm runs on $Y$ in at most $O(m\ell k^2(8/\epsilon)^D)$ time.}

    %\mjgtext{This needs more}
    
    \subsection{A Practical Heuristic}
    \label{sec:heuristic of mapreduce}
    The size of the  coreset in our algorithm can be viewed as a parameter that affects both  the 
   {memory usage} and the approximation ratio.
    %performance. 
    Until now, we focused on fixing the  worst-case approximation ratio and let that specify the memory required. %set up the parameter by the target ratio $\epsilon$. However, 
    In practice, we can deal with this parameter more flexibly. 
    %In the real-world implementation, space memory has strict restrictions. Based on this, 
    In real-world implementations, memory-space memory can be restricted.
    Inspired by a similar approach in \cite{ceccarello2019solving}, we thus slightly modify our algorithm and use permitted memory size itself as an input, instead of the approximation ratio.
    
  {Our new  algorithm (heuristic)  will start by restricting the size of the coreset to some given value $Q$
    %Suppose that we aim to constuct a coreset with size $Q$ at each processor 
    (w.l.o.g., assume $Q>k$).
    We now describe the procedure and also show that this coreset becomes an $\epsilon$-coreset when $Q$ is large enough. }
    
    %This practical algorithm constructs a coreset whose size is exactly $Q$. Finally, we show that this coreset becomes an $\epsilon$-coreset when $Q$ is large enough.
 {This new algorithm is two phases but is even easier to implement.}
    %idea of this algorithm is very simple to implement.
    During the first phase, after receiving the point set $P_t$, each processor $t$ uses the $2$-approximation greedy algorithm  from \cite{gonzalez1985clustering} to solve the $Q$-center problem on $P_t$. This generates a solution set $Y_t$ of $Q$ points. Each point $p\in P_t$ is then assigned to  its closest point $y^{t}\in Y_t$. All points that are assigned to the same center $y_j^t\in Y_t$ form a cluster $X_j^t$. By definition, $\mathcal{T}_t=\{X_j^t:y_j^t\in Y_t\}$ is a partition of $ P_t$. 
   {We then, as in   definition \ref{def:coreset_of_faircenter},  construct  vector $\COL(y)$ and set $\Pot(y)$ for each $y\in Y_t.$}

   %%\mjgtext{as in}  definition \ref{def:coreset_of_faircenter}, after  defining vector $\COL(y)$ and set $\Pot(y)$ for each $y\in Y_t,$  %and $1\leq i\leq m$, 
   % $Y_t$ is an $\epsilon$ coreset of $P_t$ for the fair $k$ center problem.  This coreset is the same as the one constructed in \cref{ALg:Map}.

    Each processor $t$  then sends $Y_t$ along with the associated  vector $\COL(y)$ and sets $\Pot(y)$ for all $y\in Y_t$, to the coordinator. 
    { $Y_t$ is a solution of the $Q$-center problem, so $|Y_t|\le Q.$}

    The process concludes by having the coordinator directly run the JNN algorithm on $\bigcup_t \bigcup_{t \in Y_t} \Pot(y)$ to construct a feasible solution.
    %and $|\Pot(y)|\leq m,$ %$\forall y\in Y_t$, 
    %processor $t$ sends at most $mQ$ information to the coordinator, while the coordinator receives $mQ\ell$ information. 

    %\mjgtext{I though the local memory is only $O(mQ)$.  How can the processor handle $\ell m Q$ data?}

    %\Cref{ALg:Coordinator} shows that there exists $Y\subseteq \bigcup_{1\leq t\leq \ell}Y_t$ such that $Y$ is an $\epsilon$ coreset of $P$ when for each $t$ $Y_t$ is an $\epsilon$ coreset of $P_i$. For a practical application, we directly run JNN algorithm \cite{jones2020fair} on $\bigcup_t \bigcup_{t \in Y_t} \Pot(y)$ to construct a feasible solution.

    %\mjgtext{Something is missing here.  You can't run JNN on $\bigcup Y_t$ directly. $\bigcup Y_t$  night not include a feasible set.  Do you first use something like \cref{lem:approximation_in_coreset}?  Or do you run JNN on $\bigcup_t \bigcup_{t \in Y_t} \Pot(y)?$
    }
    %$\forall y^{t}\in Y_t$, if there exists a point $x\in X_j^t$ from group $i$, we set $COLOR_i(y^t)=1$

    %Recall that $r_{kC}(P_i)$ and $r_{\beta C}(P_2)$ denote the optimal value of the $k$ center problem and the $\beta$ center problem on $P$ respectively. Since $Y_i$ is generated by the 2 approximation algorithm algorithm \cite{gonzalez1985clustering}, we have $r_\beta^i \leqslant 2r_{\beta C}(P_i)$. We know $\beta>k$, so $r_{\beta C}(P_i)\leqslant r_{kC}(P_i)$, i.e., $r_\beta^i\leqslant 2r_{kC}(P_i)$

    \begin{theorem}\label{thm:heuristic_mapreduce}
        When $Q$ is large enough, the heuristic is a  $3(1+\epsilon)$ approximation algorithm. (see proof in appendix)
    \end{theorem}

    %\mjgtext{QUESTION: IF $Q$ IS LARGE ENOUGH IS THIS GUARANTEED TO BE  AS GOOD AS $(3+\epsilon)$?}

%% file: streaming.tex
\section{The Dynamic/Streaming Setting}
    %As discussed in \cref{sec:coreset}, if we knew a lower bound $r$ of the optimal value in advance, it would easy to maintain an $\epsilon$-proper $(\epsilon r,\alpha)$ net of $P$ in the streaming setting. 
    %In the streaming setting, 
    For $t\leqslant n=|P|$, let $P(t)$   denote the set of first $t$ points read and $r^*(t)$ the optimal value of  $FkCP$ on $P_t.$

    \subsection{Robust Setting}
    In order to use our techniques in the streaming setting we will need  a lower  bound on $r^*(t)$.
    Such a bound already exists.
     %In our streaming algorithm, we use 
     More specifically, \cite{charikar2004incremental} provide an incremental algorithm that maintains such a lower bound $r(t)$ of $r^*(t)$. Their algorithm 
     %In \cite{charikar2004incremental}, authors 
     actually maintains a solution set $S(t)$,  $|S(t)|\leqslant k$ such that (1) $P(t) \subset\bigcup_{s\in S(t)}B(s,8r(t))$; (2) $\forall s_1,s_2\in S(t)$ $d(s_1,s_t)>4r(t)$; (3) $\forall t$ $r(t)\leqslant r^*(t)$; and (4) $r(t+1)=2^\lambda r(t)$ where $\lambda$ is a non-negative integer and computed by the incremental algorithm.
    
    %When they read $t+1$th point $p_{t+1}$, they consider the following cases:
    %\begin{itemize}
     %   \item if there exists $s\in S(t)$ such that $d(p_{t+1},s)\leqslant 8r(t)$,  $S(t+1)=S(t)$ and $r(t+1)=r(t)$.
      %  \item if $\forall s\in S(t)$ satisfying $d(p_{t+1},s)> 8r(t)$ and $|S(t)|<k$, $S(t+1)=S(t)\cup \{p_{t+1}\}$ while $r(t+1)=r(t)$.
       % \item if $\forall s\in S(t)$ satisfying $d(p_{t+1},s)> 8r(t)$ and $|S(t)|=k$, we need to update the lower bound. Let $S'(\lambda)$ be the maximal subset of  $S(t)\cup \{p_{t+1}\}$ such that $\forall s_1,s_2\in S'(\lambda)$ $d(s_1,s_2)> 4\cdot 2^\lambda r(t)$. Noted that $|S'(0)|=|S(t)\cup \{p_{t+1}\}|=k+1$. Compute the smallest integer $\lambda$ such that $|S'(\lambda)|\leqslant k$. Then, we set $S(t+1)=S'(\lambda)$ and $r(t+1)=2^\lambda r(t)$ 
    %\end{itemize}
    
    When we use this incremental algorithm as a subroutine, robustly set a target ratio $3(1+\epsilon)$ and define $\Bar{\epsilon}=\epsilon/3$, we can incrementally maintain an $\Bar{\epsilon}$-proper $(r,2)$ net $Y$ of $P$.

    %\mjgtext{NEED TO EXPLAIN HOW TO USE \cref{algo:coreset} ON LINE 8 TO CONSTRUCT THE NEW NET. THAT GOT LOST SOMEHOW}
    \begin{algorithm}[h!]
	\caption{Streaming algorithm for constructing an {$(\Bar{\epsilon} r,2)$} net $Y$ of $P$} 
	\label{ALg:Streaming}
	%\hspace*{0.02in} 
	{\bf Input:} Ordered set $P=\{p_1,...,p_n\}$
	\begin{algorithmic}[1]
	\STATE $Y(0)=\emptyset$
    \STATE When $p_t$ is read
    \IF{$t\leqslant k$}
    \STATE $Y(t)=Y(t-1)\cup \{p_t\}$
    \ELSE
    \STATE Apply  \cite{charikar2004incremental}'s incremental algorithm to calculate lower bound $r(t)$ of $r^*(t)$.
    \IF{$r(t)>r(t-1)$}
    \STATE { Apply \cref{algo:coreset} with $Y_1=Y(t-1)$ and $Y_2=\emptyset$ to construct a new $(\frac{\Bar{\epsilon} r(t)}{2},2)$ net $Y(t-1)$ of first $t-1$ points.}
    \ENDIF
    \IF{there exists $y\in Y(t-1)$ such that $d(p_t,y)\leqslant \Bar{\epsilon} r(t)$}
    \IF{$\COL_{g(p_t)}(y)=0$}
    \STATE $\COL_{g(p_t)}(y)=1$ and $\Pot(y,g(p_t))=p(t)$
    \STATE $Y(t)=Y(t-1)$
    \ENDIF
    \ELSE
    \STATE $Y(t)=Y(t-1)\cup \{p_t\}$
    \STATE  $\COL_{g(p_t)}(p(t))=1$ and $\Pot(p_t,g(p_t))=p_t$
    \STATE $\COL_i(p_t)=0$ $(\forall i\neq g(p_t))$
    \ENDIF
    \ENDIF
	\end{algorithmic}
    \end{algorithm}	
    
    \begin{lemma} \label{lem:coreset_of_Streaming}
        \Cref{ALg:Streaming} computes an $\Bar{\epsilon}$-proper $(\frac{\Bar{\epsilon}}{2} r(t),2)$-net $Y_t$ of  $P_t$, where $|Y_t|=O(k(32/\Bar{\epsilon})^D)$. (see proof in appendix)
    \end{lemma}

    Finally, similar to the previous section, 
    conclude by using \cref{lem:approximation_in_coreset} with $\bar \epsilon = \epsilon/3$ and  calling the $3$-approximation JNN algorithm \cite{jones2020fair} for the $k$-center problem. This returns a  $3(1+\epsilon)$ approximation solution in at most $O(m k^2(96/\Bar{\epsilon})^D)$ time. 

    %\mjgtext{Please check out the paragraph above and "fix" the $3 \epsilon$, $\epsilon$ issue.}

    %compute the fair $k$ center problem by running JNN algorithm \cite{jones2020fair} on this $(\frac{\epsilon }{2}r(t),2)$-net $Y(t)$. By lemma \ref{lem:approximation_in_coreset}, JNN algorithm returns a $3(1+3\epsilon)$ approximation solution in at most $O(m k^2(32/\epsilon)^D)$ time.

    \subsection{A Practical Heuristic}
    As in \cref{sec:heuristic of mapreduce}, we slightly modify our algorithm and use memory  space instead of the target  approximation-ratio as an input parameter. 

    Again as in \cref{sec:heuristic of mapreduce}, 
    our new  algorithm (heuristic)  will start by restricting the size of the coreset to some given value $Q$ (w.l.o.g., assume $Q>k$). %Finally, we show that this coreset becomes an $\epsilon$-coreset when $Q$ is large enough.
    
    We next directly apply the incremental algorithm \cite{charikar2004incremental} to solve the $Q$-center problem on the data stream. Different from \cite{charikar2004incremental}, each center $x$ will now have  an associated $\COL(x)$ function and a set $\Pot(x)$. 
    {At each step the algorithm also updates the group information associated with this $Q$-center. }
    %{\color{blue} When we read a new point or double the lower bound $r(t)$, group information requires to maintenance}.  
    Due to space limitations, we describe the details of the heuristic algorithm in the appendix.
%\mjgtext{PARAGRAPH ABOVE NEEDS TO BE FIXED TO PARALLEL Section 4.2. IN 4.2 CORESET SIZE IS THE PARAMETER. HERE ITS MEMEORY, WHICH IS SLIGHTLY DIFFERENT
    %} 

%1\mjgtext{YOU SHOULD AT LEAST GIVE A ONE LINE INTUITION HERE AS TO WHAT THE ALGORITHM IS DOING. YOU MAKE IT SOUND AS IF ALL YOU ARE DOINGB IS MAINTAINING THE $S(t)$ FROM cite{charikar2004incremental} but you are doing more than that}

    Our heuristic algorithm then runs JNN algorithm \cite{jones2020fair} on the coreset constructed to generate a feasible solution.  As in \cref{sec:heuristic of mapreduce}, we show that for large enough $Q$ the heuristic is a $(3+\epsilon)$ approximation algorithm.

    \begin{theorem}\label{thm:heuristic_streaming}
        When $Q$ is large enough, the heuristic is a  $3(1+\epsilon)$ approximation algorithm.  (see proof in appendix)
    \end{theorem}

  %  \mjgtext{PLEASE FIX THE PARAGRAPHS ABOVE SO THEY CLEARLY PARALLEL 4.2}

%% file: experiment.tex
\section{Experiments}
In this section, we run experiments to evaluate the performance of our heuristic one-pass and MapReduce algorithm on some real-world datasets and a massive synthetic dataset. Though the theoretical guarantee of $3+\epsilon$ for both algorithms requires the low dimensionality condition, {i.e., bounded doubling dimension,} condition, the results are still very good for the high dimensional datasets that do not satisfy those conditions.  The one-pass algorithm, despite being incremental,  achieves a similar performance ratio but with \emph{faster running time and lower memory usage} compared to the previous best algorithms. The MapReduce algorithm outputs the smallest cost solution 
%\mjgtext{what is best} 
in most experiments, while exhibiting \emph{a much better ratio} for the low dimensional case. 

% We totally implement five algorithms in our experiments. JNN(cite)'s offline algorithm (cite), XXX's Two Pass streaming algorithm, and our One Pass streaming algorithm; Chiplunkar's distributed algorithm, and our Map Reduce Algorithm. Refer to section xxx for a detailed illustration of the implementation. 

\subsection{Datasets}

%\mjgtext{Are these the SAME datsets as in \cite{chiplunkar2020solve}? If yes, make that explicit}

We used the same datasets and preprocessing methods as \cite{chiplunkar2020solve}, including three real world datasets: CelebA, Sushi, Adult, and one synthetic dataset: Random Euclidean. All of them use the $\ell_1$ metric with the exception of SushiA where the pairwise distance between ranking orders is calculated by the number of inverse pairs.
%inverse pairs to evaluate the distances of ranking orders. The $\ell_1$ distance is utilized as the metric for most cases, with the exception of SushiA where we count inverse pairs to evaluate the distances of ranking orders.
% Additionally, we considered the dataset Higgs, which provides an example of massive data in real world problem.

\textbf{Sushi}\cite{SushitData} contains $5\,000$ responses to a sushi preference survey. There are two types of evaluations given: \textbf{SushiA} contains the ranking order of $10$ kinds of sushi, and \textbf{SushiB} contains the score of $50$ kinds of sushi. The attributes given are gender and six age groupings; this results in 2 groups (gender) \footnote{We sincerely apologize for any offense caused by the binary classification of ``male'' and ``female'' in the group representations.}, 6 groups (age), or 12 groups ($\text{gender} \times \text{age}$).

 \textbf{Adult}\cite{AdultData}  contains $32\,561$ data points extracted  from the 1994 US Census database in which  education, occupation and other aspects are covered, and will be considered as $6$-dimensional features after normalizing. Using gender (2) and race (5) information, this generates groups 
 %like Sushi 
 with sizes 2, 5 and  10.

 \textbf{CelebA}\cite{liu2015deep}  contains $202,599$ face images %with several attributes.
  After preprocessing, the $15\,360$ dimensional features are extracted via pre-trained VGG16 using Keras, and groups are divided by gender (2 groups), or $\text{gender} \times \{\text{young}, \text{not young}\}$ (4 groups). Since it is extremely high dimensional, it can test scalability of our algorithm.

%\mjgtext{What is "young"?  This was unclear.}
 
    % \item \textbf{Higgs} contains $11\,000\,000$ Monte Carlo simulated physical signals related to Higgs bosons, which provides an example of meaningful big data. Among $21$ low-level features and $7$ high-level features, we used the high-level features as coordinates, and two groups by the classification label: signal or background noise.

 \textbf{Random Euclidean}\cite{chiplunkar2020solve} is a synthetic 100GB dataset designed by Chiplunkar et al. It  contains $4\,000\,000$ uniformly generated  points in $1\,000$-dimensional Euclidean space, each randomly assigned with to  of 4 groups. It is useful to illustrate the performance  of algorithms when input data is larger than memory.
% \paragraph{CelebA} CelebA dataset contains $202\,599$ face images with several attributes. After pre-processing, the $15\,360$ dimensional features are extracted via pretrained VGG16 using Keras, and groups are divide by gender (2 groups), or $\text{gender} \times \text{young}$ (4 groups).
% \paragraph{Sushi} Sushi contains $5\,000$ responses of sushi preference survey. There are two types of evaluations given: SushiA contains the ranking order of $10$ kinds of sushi, and SushiB contains the score of $50$ kinds of sushi. The attributes are given in gender and six age groups, result in 2 groups (gender only), 6 groups (age only), or 12 groups (combination).
% \paragraph{Adult} Adult dataset contains $32\,561$ extracted data points from 1994 US Census database where education, occupation and other aspects are covered. Using gender (2) and race (5) information, we can generate group like Sushi with sizes 2, 5, 10.
% \paragraph{Higgs} Higgs contains $11\,000\,000$ Monte Carlo simulated physical signals related to Higgs bosons, which provides an example of meaningful big data. Among $21$ low-level features and $7$ high-level features, we used the high-level features as coordinates, and two groups by the classification label: signal or background noise.
% \paragraph{Random Euclidean} Chiplunkar and Kale's synthetic 100GB dataset contains uniformly generated $4\,000\,000$ points in $1000$-dimensional Euclidean space, each assigned with one of 4 groups randomly. It is useful to illustrate the availability of algorithms when input data is larger than memory.

\subsection{Implementation Details}
The experiments were run on a PC with AMD Ryzen 7 2700X Processor @ 3.7GHz, 32GB Memory and 500GB Solid-State Drive. We used Python to implement the algorithms, and ran experiments on the several real datasets and  massive synthetic dataset {previously described.} %Our code repository is available on GitHub\footnote{\url{https://github.com/fstqwq/fair-k-center}}.

{\bf Previous algorithms} We adopted and refined Chiplunkar 
 et al. implementation\footnote{\url{https://github.com/sagark4/fair_k_center}} in order to compare algorithms \emph{fairly}: the reproduction of their results ensures our comparisons are reliable.

%\mjgtext{Below you wrote "overperform" which I replaced by outperform". Is that what you meant? Also, what did you mean by "as is"?  Is that what was in the previous paper. If so, make that explicit!}

%\mjgtext{HOW DO YOU KNOW THAT \cite{chiplunkar2020solve}'S OUTPERFORM THE OTHERS. IF THAT WAS IN THEIR PAPER, YOU SHOULD REFERENCE THAT FACT}

In this section, the streaming algorithm and the distributed algorithm from \cite{chiplunkar2020solve} are respectively labeled  as \textbf{Two Pass} and \textbf{CKR Distributed}. 
According to \cite{chiplunkar2020solve},
these two algorithms generally {outperform} \cite{chen2016matroid} and  \cite{kleindessner2019fair}.  We therefore compare  our algorithms directly to 
\cite{chiplunkar2020solve}'s algorithms, keeping the parameters the same, e.g., $\epsilon = 0.1$, as they used.
%\textbf{Two Pass} and \textbf{CKR Distributed}, keepi
%run these two algorithms for comparison, keeping parameters exactly the same as their implementation provided, e.g., $\epsilon = 0.1$.  %, and Kale's Two Pass by Chiplunkar and Kale's result.  
 We also followed their format of using the cost of the output of (their implementations of) Gonzalez's  algorithm as the \textbf{Lower Bound} that all of the other algorithms are compared to.
 %as the lower bound  (1.0) in the comparison.

 %\mjgtext{I DIDN'T CLEARLY UNDERSTAND WHAT YOU MEANT BY USING GONZALEZ AS A LOWER BOUND. PLEASE CLARIFY THIS}

{\bf Our Implementations} (1) In our implementation of the heuristic One Pass algorithm the coreset size is set to a constant $240$. This was chosen to be divisible by the number of processors and the sum of group sizes, and also to let the  two streaming algorithms use comparable memory.  (2) Our MapReduce algorithm is implemented by a \texttt{multiprocessing} library on a single machine. The number of processors is set to $10$ to fit the CPU capacity. {For the first three datasets the size of the coreset collected by the coordinator is the same as in One Pass ($240$), but for 
%\mjgtext{what does mostly mean here? This ounds bad and needs tio be rewritten}; except for 
Random Euclidean, we used $800$ as a coreset size to better utilize the simulated $100$ processors, where the number is chosen so that two distributed algorithms will send exactly the same number (3\,200) of points to the coordinator.
%(3) We wrote our own implementation of   the \textbf{JNN} algorithm as described in \cite{jones2020fair}.

\subsection{Results}

 To evaluate the scalability of streaming algorithms, we use the first $32\,500$ points in dataset Adult;  each group was allowed at most 10 centers (denoted by capacities $[10, 10]$, i.e., $10$ men and $10$ women).  We require the algorithm to report a solution after completing reading a multiple of $2\,500$ points.

% \mjgtext{I REWROTE THE PARAGRAPH BELOW. PLEASE CHECK IT}
 %{\color{blue}each 2\,500 points; 
 %\mjgtext{
 Note that  One Pass is updating the coreset  after reading each point so far, after reading a multiple of $2\,500$ points and reporting an approximate  $k$-center solution using the JNN \footnote{We write our implementation for JNN because we fail to find JNN's original implementation.This calls  a maxflow subroutine from  the \texttt{networkx} library.} algorithm , it can continue with the new points without having to backtrack and reprocess the old ones again.
 The reported running time of the One Pass  at each checkpoint is then just the time to update the coresets and then to calculate the approximate $k$-centers.
 By comparison, Two Pass has to rerun the algorithm on the whole data set from the scratch. 
To make the comparisons between the algorithms more realistic, we also calculate the entire running time of One Pass  if it started from scratch on the dataset up to that point.

 %Recall that, when new 2\,500 points come in, our one pass algorithm only needs to read these new points and update the current coreset of the previous data set without rereading the previous points,  while the two pass algorithm has to rerun the algorithm on the whole data set from the scratch. We record that running time for each checkpoint; time records for One Pass are the update (update the coreset when new 2500 points are inserted) and query time (run JNN on the coreset to generate a feasible solution) from the last checkpoint.}  \mjgtext{What does capacities mean here? Do you mean that you are using gender? Also, why did you choose 10,10?  What does it mean to ask the algorithm to report a solution each $2\,500$? Are you doing that for two Pass as well? What does that mean?}
 The results are   shown in Figure \ref{fig:exp_1}. As input size grows, the two algorithms have similar solution quality when One Pass is set to use only half of Two Pass's memory. Meanwhile, One Pass shows a significantly higher efficiency over Two Pass since it can incrementally maintain coresets and obtain solutions upon request anytime. It is  worth noting that One Pass remains faster  than Two Pass even if  it is required to run from scratch.
 
 %{\color{blue}(compute the solution on the whole data set instead of only reading new 2500 points).} \mjgtext{What does "from scratch" mean here?}

\begin{figure}[!ht]
    \centering
    \includegraphics[width=70mm]{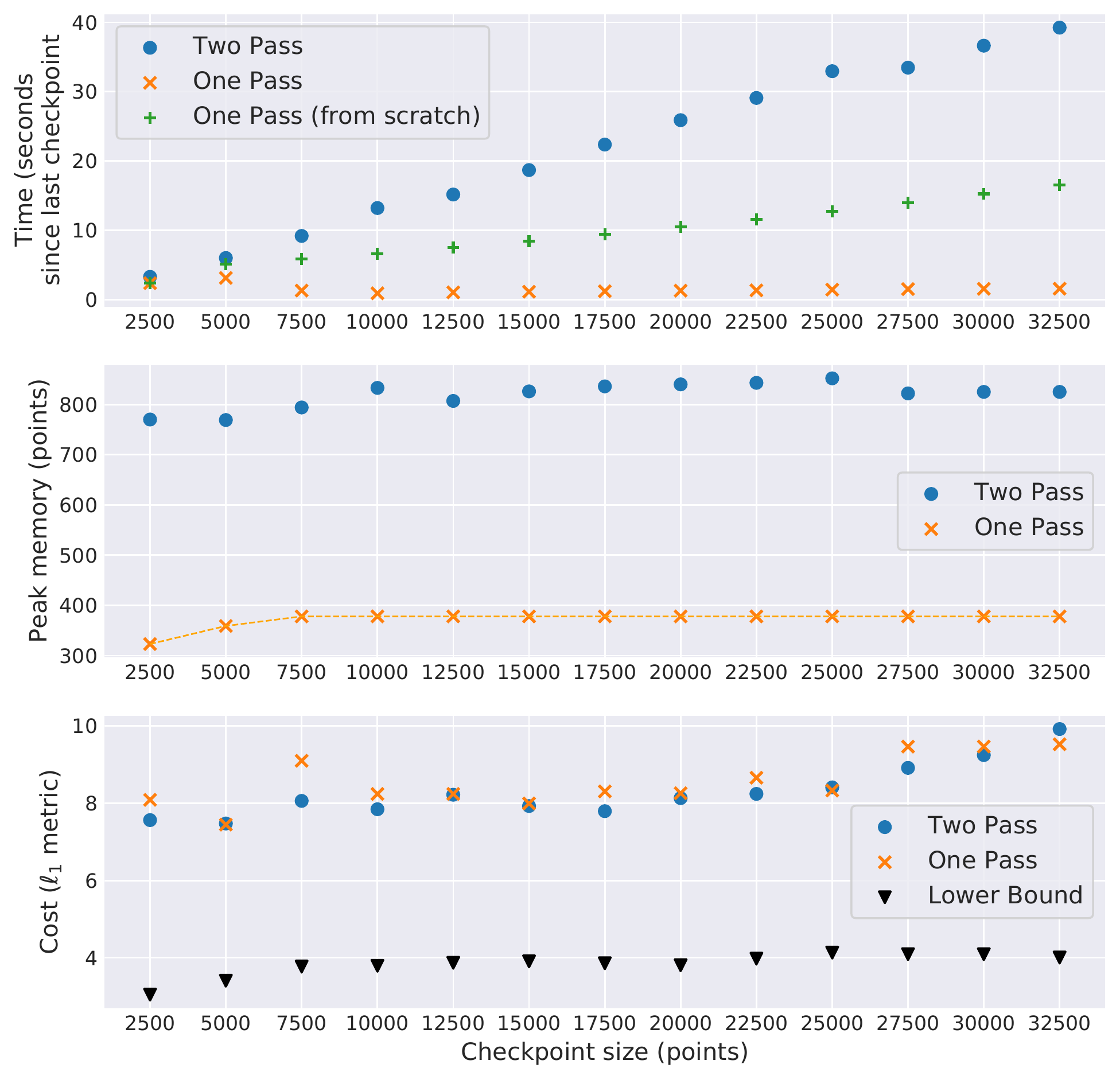}
    \caption{Checkpoint running comparison for dataset Adult with capacities: [10, 10]. Time is the average of 3 runs.}
    \label{fig:exp_1}
\end{figure}

%\mjgtext{I didn't understand the paragraph below. Please be more precise. I can clean it up after you write the details}

We then ran experiments on all of the datasets.
%Then, we run experiments on {\color{blue} all} datasets for a precise numerical result.

\begin{table}[!ht]
    \centering\scalebox{0.6}{
    \begin{tabular}{|c|c|c?c|c|c?c|c|}
    \hline
        \multirow{2}{*}{Dataset} & \multirow{2}{*}{Size} & \multirow{2}{*}{Capacities} & \multicolumn{3}{c?}{Time (seconds)} & \multicolumn{2}{c|}{Memory (points)} \\
        \cline{4-8}
        & & & JNN & Two Pass& One Pass & Two Pass & One Pass \\ \hline
        % Dataset & \makecell{Number \\of points} & Capacities & \makecell{JNN\\Time} & \makecell{Two Pass\\Time} &\makecell{One Pass\\Time} & \makecell{Two Pass\\Memory} &\makecell{One Pass\\Memory} \\ \hline
        \multirow{3}{*}{SushiA} & \multirow{3}{*}{5000} & [10, 10] & 12.70 & 5.88 & \bf 3.92 & 464 & \bf 429 \\ \cline{3-8}
         &  & [10] * 6  & 43.41 & 16.13 & \bf 14.24 & 1766 & \bf 933 \\ \cline{3-8}
         &  & [5] * 12  & 41.44 & \bf 15.96 & 17.07 & 2134 & \bf 1489 \\ \hline
        \multirow{3}{*}{SushiB} & \multirow{3}{*}{5000} & [10, 10] & 8.32 & \bf 2.37 & 2.62 & \bf 230 & 265 \\ \cline{3-8}
         &  & [10] * 6 & 14.6 & 7.21 & \bf 2.00 & 789 & \bf 733 \\ \cline{3-8}
         &  & [5] * 12 & 11.92 & 7.20 & \bf 1.87 & \bf 838 & 969 \\ \hline
        \multirow{3}{*}{Adult} & \multirow{3}{*}{32560} & [10, 10] & 57.88 & 39.36 & \bf 16.69 & 825 & \bf 378 \\ \cline{3-8}
         &  & [10] * 5 & 114.1 & 84.34 & \bf 30.06 & 2516 & \bf 573 \\ \cline{3-8}
         &  & [5] * 10 & 113.1 & 84.75 & \bf 30.05 & 2931 & \bf 948 \\ \hline
        \multirow{1}{*}{CelebA} & 202590 & [10, 10] & 2052 & 1350 & \bf 501.2 & 516 & \bf 431 \\ \hline
         \makecell{Random\\Euclidean} & 4e6 & [2] * 4 & -- & 5191 & \bf 1383 & \bf 52 & 314 \\ \hline
    \end{tabular}}
    \caption{\label{tab:streaming}Time and memory for streaming algorithm on all datasets. JNN algorithm could not finish in a reasonable time for Random Euclidean.}
\end{table}

\begin{table}[!ht]
    \centering\scalebox{0.6}{
    \begin{tabular}{|c|c|c?c|c|c?c|c|}
    \hline
        Dataset & Capacities & \makecell{Lower\\Bound} & JNN & Two Pass & One Pass & \makecell{CKR\\Distributed}  & \makecell{Map\\Reduce} \\ \hline
        \multirow{3}{*}{SushiA}  &  [10, 10]  &  8.00  & \cellcolor{black!25} 2.00 & \cellcolor{black!0} 2.38 & \cellcolor{black!7} 2.12 & \cellcolor{black!0} 2.50 & \cellcolor{black!7} 2.12 \\ \cline{2-8}
     &    [10] * 6  &  6.50  & \cellcolor{black!25} 2.15 & \cellcolor{black!1} 2.46 & \cellcolor{black!1} 2.46 & \cellcolor{black!0} 2.62 & \cellcolor{black!25} 2.15 \\  \cline{2-8}
     &    [5] * 12  &  6.50  & \cellcolor{black!5} 2.31 & \cellcolor{black!5} 2.31 & \cellcolor{black!25} 2.15 & \cellcolor{black!0} 2.77 & \cellcolor{black!25} 2.15 \\ \hline
\multirow{3}{*}{SushiB}  &  [10, 10]  &  34.00  & \cellcolor{black!0} 2.03 & \cellcolor{black!5} 1.85 & \cellcolor{black!25} 1.71 & \cellcolor{black!2} 1.94 & \cellcolor{black!0} 2.06 \\  \cline{2-8}
     &    [10] * 6  &  30.50  & \cellcolor{black!25} 1.93 & \cellcolor{black!16} 1.97 & \cellcolor{black!6} 2.07 & \cellcolor{black!16} 1.97 & \cellcolor{black!25} 1.93 \\  \cline{2-8}
     &    [5] * 12  &  30.50  & \cellcolor{black!25} 1.97 & \cellcolor{black!25} 1.97 & \cellcolor{black!25} 1.97 & \cellcolor{black!25} 1.97 & \cellcolor{black!25} 1.97 \\ \hline
\multirow{3}{*}{Adult}  &  [10, 10]  &  4.01  & \cellcolor{black!25} 2.08 & \cellcolor{black!1} 2.41 & \cellcolor{black!1} 2.38 & \cellcolor{black!0} 2.78 & \cellcolor{black!17} 2.12 \\  \cline{2-8}
     &     [10] * 5  &  3.04  & \cellcolor{black!25} 2.45 & \cellcolor{black!12} 2.54 & \cellcolor{black!9} 2.57 & \cellcolor{black!0} 2.93 & \cellcolor{black!15} 2.51 \\  \cline{2-8}
     &    [5] * 10  &  3.04  & \cellcolor{black!23} 2.45 & \cellcolor{black!3} 2.71 & \cellcolor{black!0} 2.93 & \cellcolor{black!3} 2.68 & \cellcolor{black!25} 2.44 \\ \hline
CelebA  &  [10, 10]  &  40796  & \cellcolor{black!10} 1.89 & \cellcolor{black!3} 1.99 & \cellcolor{black!3} 2.00 & \cellcolor{black!8} 1.91 & \cellcolor{black!25} 1.81 \\ \hline
\makecell{Random\\Euclidean} &  [2] * 4  &  --  &  --  & \cellcolor{black!24} 3.454e7 & \cellcolor{black!25} 3.450e7 & \cellcolor{black!21} 3.475e7 & \cellcolor{black!23} 3.461e7 \\ \hline
    \end{tabular}}
    \caption{\label{tab:large}Costs on all datasets. Each column after the third corresponds to an algorithm and shows the ratio of its cost and Gonzalez’s lower bound. The shaded values indicate the ratios to Lower Bound if available, darker is better.}
\end{table}

 Table \ref{tab:streaming} compares the time and memory used by the  streaming algorithms on the different datasets.   To further contrast their efficiency,  we also listed the time used by JNN algorithm, which  %considered as the fastest in the comparison of \cite{jones2020fair} and 
used $O(n)$ memory to achieve current performance: it consumed $24$ GB memory to store points when running CelebA dataset. 

%\mjgtext{I DID NOT UNDERSTAND THE COMMENTS ABOUT JNN IN THE PARAGRAPH ABOVE. PLEASE CLARIFY. ONE QUESTION IS, WHAT DO YOU MEAN BY STEAMING JNN?}

{The two streaming algorithms are both much faster than JNN, and One Pass is much faster than Two Pass for large data. We remark that in the massive case, i.e., the Random Euclidean with 4\,000\,000 points experiments, our One Pass only  requires 23 minutes, while just processing the input points needs 21.8 minutes.} It's also noticeable that the One Pass algorithm can better utilize given memory.
The memory usage of Two Pass highly depends on the aspect ratio  $\Delta$ of the data set; it uses little memory on the Random Euclidean dataset since its $\Delta$ is quite small (about 2.16) and uses much more memory for larger $\Delta$ in the other, real, datasets. Comparatively, One Pass is  more adaptive to a fixed given coreset size.

%\mjgtext{WHAT DID YOU MEAN BY GEOMETRIC GUESSING IN THE PARAGRAPH ABOVE? THAT"S CONFUSING}

In the middle three columns of Table \ref{tab:large}, we compare the costs of different single-threaded algorithms having similar theoretical guarantees%: recall that  the JNN algorithm is $3$ approximate; One Pass and Two Pass have the worst case ratio guarantee $3+\epsilon$ respectively under the low dimensional hypothesis. 
We also observe a similar experimental performance for them, while JNN usually generates the smallest cost solution.

%\paragraph{Distributed Algorithms: CK Distributed vs. Map Reduce}
The last two columns of Table \ref{tab:large} compare the  two distributed algorithms. %According to \cite{chiplunkar2020solve}, CKR Distributed exhibits   good experimental performance even though its theoretical guarantee is only  $17+\epsilon$. Our MapReduce algorithm performs best on most datasets, while its worst case ratio guarantee is  $3+\epsilon$ under the low dimensional condition. 
Both algorithms are fast: MapReduce took $23$ minutes on Random Euclidean and CKR Distributed took $27$ minutes. We do not compare the precise timing results for these two distributed algorithms, as we did not simulate the IO process in a realistic distributed environment. Therefore, the running time in our experiment may not provide much insight about the efficiency of the two algorithms in a real-world setting.

%\mjgtext{The last few sentences in the above were confusing}

%% file: appendix.tex
\section{Some Proofs}

{\bf Proof of \cref{lem:net}}
  \begin{proof}
    Let $S$ denote a feasible set and $r=\max_{p\in P} d(p,S)$.  Since $r^*_{FkC}(P)$ is optimal, $r'\leqslant r^*_{FkC}(P)\leqslant r$.  Let $p_0=\arg \max_{p\in P} d(p,S).$ %and $y_0=\arg \max_{y\in Y} d(y,S)$. 
    From the covering  property of the $(\frac{\epsilon}{\alpha} r',\alpha)$-net, there exists a point $y'\in Y$ %and $p'\in P$
    such that $d(y',p_0)\leqslant \epsilon r'$. %and $d(p',y_0)\leqslant \epsilon r'$. 
    
    From one direction,  since $Y\subseteq P$, $\max_{y\in Y}d(y,S)\leqslant \max_{p\in P}d(p,S)\leqslant (1+\epsilon)r$
       % $$\begin{aligned}
        %\max_{y\in Y}d(y,S)&= d(y_0,S)\leqslant d(y_0,s_1) \\
         %       &\leqslant d(p',s_1)+d(p',y_0)\\
         %       &\leqslant d(p',S)+\epsilon r'\\
          %      &\leqslant \max_{p\in P}d(p,S)+ \epsilon r= (1+\epsilon)r
        %\end{aligned}$$ 

    From the other,  let $s_2=\arg \min_{s\in S}d(y',s)$. Then
        $$\begin{aligned}
            \max_{y\in Y}d(y,S)&\geqslant d(y',S)=d(y',s_2)\\
            &\geqslant d(p_0,s_2)-d(y',p_0) \\
            &\geqslant  d(p_0,S)-\epsilon r' \\
            &\geqslant r-\epsilon r=(1-\epsilon)r
            \end{aligned}$$ 
    Thus, $(1-\epsilon)r \leqslant \max_{p\in Y}d(p,S)\leqslant (1+\epsilon)r$ and  $Y$ is an $\epsilon$-coreset of $P.$
    \end{proof}

{\bf Proof of \cref{lem:CandHelper}}
   \begin{proof}
   Fix $p \in P.$
   %For each $ p\in P$, 
   Since $Y$ is $\epsilon$-proper, $\exists y\in Y$ such that $d(y,p)\leqslant \alpha r\leqslant \epsilon r^*_{FkC}(P)$. Since $S$ is a candidate solution with cost $c,$
   %with at most $\phi\cdot r^*_{FkC}(P)$ cost, 
   %$\exists s\in S$ $d(y,S)=d(y,s)\leqslant \phi r^*_{FkC}(P)$. 
   $\exists s\in S,$ 
   $d(y,S)=d(y,s)\leqslant c$. 
   
   By \cref{def:candidate}, %there exists a feasible set
   $S'\subseteq P$ such that $|S'\cap N(s,r)|=1$, i.e., $\exists s'\in S'$ $d(s,s')\leqslant \alpha r\leqslant \epsilon r^*_{FkC}(P)$.Thus
    %$$d(p,S')\leqslant d(p,y)+d(y,s)+d(s,s')\leqslant (\phi+2\epsilon)r^*_{FkC}(P)$$
     $$d(p,S')\leqslant d(p,y)+d(y,s)+d(s,s')\leqslant c + 2 \epsilon r^*_{FkC}(P).$$
   \end{proof}

{\bf Proof of \cref{lem:coreset_of_MapReduce}}
    \begin{proof} 
    %\mjgtext{I rewrote this proof. Please check} 
        Lines 2-5 of \cref{ALg:Map} is just the classical greedy $kPC$ algorithm from \cite{gonzalez1985clustering}; this is  known to compute a 2-approximate $kCP$ solution. Since the optimal value of $FkCP$ is no greater than the optimal value of $kCP$,  $r_t$ computed in line 6 is at most $\frac{1}{4}r^*_{FkC}(P_t)$. i.e., $\Bar{\epsilon} r_t\leqslant \frac{\Bar{\epsilon}}{4}r^*_{FkC}(P_t)$.
        
        %Lines 6 and 7 select all points in $S_t$ as parts of $Y_t$ and $d(p_{j_1},p_{j_2})\geqslant 2r_t>\epsilon r_t$. 
        %Then, line 8 read each point $p$ in $P_t$, if there exists a point $y\in Y_t$ such that $d(p,y)\leqslant \epsilon r_t$, we can add $p$ into $N(y,\epsilon r_t)$ and update $\COL_g(p)(y)$ $\Pot(y,g(p))$. Otherwise, $p$ is added into $Y_t$ and define $N(p)=\{p\}$. 

        It is straightforward to see that $Y_t$ is  an $(\Bar{\epsilon} r_t,2)$-net of $P_t$ with the $\COL(y)$ and $\Pot(y)$ vectors and $N(y,\Bar{\epsilon} r_t)$ sets correctly constructed for all $y \in Y_t$ with the $N(y,\Bar{\epsilon} r_t)$ sets forming a partition of $P_t.$

       % Overall, $Y_t$ is an $(\epsilon r_t,2)$-net of $P_t$.
       Since $\Bar{\epsilon} r_t\leqslant \frac{\Bar{\epsilon}}{4}r^*_{FkC}(P_t),$  %(2) $\bigcup_{y\in Y_t} N(y)=P_t$ (3)$\forall x,y\in Y_t$ $N(x)\cap N(y)=\emptyset$, 
       $Y_t$ is proper.
        
        Finally, %by the properties of the greedy algorithm, 
        by definition of the greedy algorithm
        all points in $P_t$ can be covered by {$\bigcup_{j= 1}^k B(p_j^t,8r_t)$,} and
        %\mjgtext{WHY?  Shouldn't this be $B(p_j^t,4r_t)$ } 
       % and for each $j$ we have 
        $$\forall j,\quad \max_{ x,y\in B(p_j^t,4r_t)\cap Y_t}d(x,y)\leqslant 16r_t.$$
        In addition, from the condition in line 10, $\forall x,y\in Y_t$ $d(x,y)\geqslant 2\Bar{\epsilon} r_t$. The aspect ratio of $Y_t\cap B(p_j^t,8r_t)$ is thus  at most $\frac{16r_t}{2\Bar{\epsilon} r_t}\leqslant \frac{8}{\Bar{\epsilon}}$. By \cref{lem:doubling}, $|Y_t\cap B(p_j^t,8r_t)|\leqslant O((\frac{8}{\Bar{\epsilon}})^{O(D)})$. Therefore, $|Y_t|=O(k(8/\Bar{\epsilon})^D)$.
    \end{proof}

{\bf Proof of \cref{lem:property_of_greedy_algo}}
     \begin{proof}
         The proof is by contradiction. Assume $A(Q)> 2\cdot r^*_{kPC}(P)$.  Then,  $\exists q\in Q$, $d(q,S)> 2\cdot r^*_{kPC}(P)$. 
         
         Since $S$ is constructed by a greedy procedure,  $\forall s_1,s_2\in S$ $d(s_1,s_2)\geqslant d(q,S)>2\cdot r^*_{kPC}(P)$. Hence, any two points $q_1,q_2$ in $\{q\}\cup S$ satisfy  $d(q_1,q_2)>2\cdot r^*_{kPC}(P)$.
         %i.e., there exist at least $k+1$ points at distance greater than $2\cdot r^*_{kPC}(P)$ from each other. 
         %We know that $S\cup \{q\}\subset Q\subset P$.
         Since $|S\cup \{q\}|=k+1,$
         in the optimal $k$-center solution of $P$, there must exist at least one center $o$ covering two points $q_1,q_2\in S\cup \{q\}$. Therefore, by the triangle inequality, we reach the contradiction 
         $$2\cdot r^*_{kPC}(P)<d(q_1,q_2)\leqslant d(q_1,o)+d(o,q_2)\leqslant 2\cdot r^*_{kPC}(P).$$
     \end{proof}

 {\bf Proof of \cref{lem:coreset_of_coordinator}}
 \begin{proof}
     Since $P_t\subset P $, by  \cref{lem:property_of_greedy_algo}, every  $r_t$ sent to the coordinator is no greater than $\frac{1}{4} r^*_{kCP}(P)$. Hence, by definition of $R$ in \cref{ALg:Coordinator}, $R\leqslant \frac{1}{2}r^*_{kCP}$. Thus, by \cref{lem:update_r}, $Y$ is an $\Bar{\epsilon}$-proper $(\frac{\Bar{\epsilon}}{2} R,2)$ net $Y$ of $P$. 

     Recall that running time of \cref{algo:coreset} is $O(m|Y_1|\cdot |Y_1\cup Y_2|)$. Therefore, the running time of \cref{ALg:Coordinator} is at most $\sum_{i=1}^\ell O(m |Y_i|\cdot |\bigcup_{j=1}^\ell Y_j|)=O(m\ell k^2 (8/\Bar{\epsilon})^{2D})$

    %Finally, by definition of $R$,  by the property of the greedy algorithm, all points in $P_t$ can be covered by {\color{blue}$\bigcup_{j= 1}^k B(p_ j^t,4r_t)$,} and for each $j$ we have $$\max_{ x,y\in B(p_j^t,2r_t)\cap Y_t}d(x,y)\leqslant 4r_t.$$ Similar to the proof in \cref{lem:coreset_of_MapReduce}, we can show that $|Y|=O(k(4/\epsilon)^D)$

    \end{proof}
    
{\bf Proof of \cref{thm:heuristic_mapreduce}}
    \begin{proof}
        Define $r_Q^t=\max_{p\in P_t}d(p,Y_t)$. Recall that $r^*_{FkC}(P_t)$ denotes the optimal value of the fair $k$ center problem on point set $P_t$. $r_Q^t$ is monotone decreasing  as $Q$ increases. Note that once $r_Q^t\leqslant \frac{\epsilon}{2}r^*_{FkC}(P_t)$, i.e., the coreset size $Q$ is large enough, by definition \ref{def:proper}, $Y_t$ is an $\epsilon$-proper $(\frac{r_Q^t}{2},2)$ net of $P_t$.
        {Thus, for large enough $Q$, $Y_t$ is an $\epsilon$ coreset of $P_t$ for the fair $k$ center problem.  This coreset is the same as the one constructed in \cref{ALg:Map}.}
        
        Let $O=\{o_1,...,o_k\}$ denote the optimal solution of $FkCP(P).$ Suppose that $\forall o_j\in O$ is assigned to the processor $t$, by definition of $Y_t$, $o_j\in N(y(o_j),r_Q^t)$ for some $y(o_j) \in Y_t.$ Hence, $o_j$ is from group $i$ so there exists \Pot$(y(o_j),i)$ (denoted by $o_j'$) that is sent to the coordinator such that 
       $$d(o_j',o_j)\leqslant d(o_j,y(o_j))+d(y(o_j),o_j')\leqslant 2r_Q^t.$$ 

       Let $r_Q=\max_{1\leqslant j\leqslant \ell}r_Q^j$. Thus, $\{o_1',...,o_k'\}\subseteq \bigcup_t \bigcup_{t \in Y_t} \Pot(y)$ is a feasible solution   with at most $(r^*_{FkPC}(P)+2r_Q)$ cost, which can cover all points in $P$. 

       Our heuristic runs JNN on $\bigcup_t \bigcup_{t \in Y_t} \Pot(y)$ and return a solution $S$ with cost $C(S)$.
       By the property of JNN \cite{jones2020fair}, 
       $$C(S)\leqslant 3\cdot r^*_{FkC}(\bigcup_t \bigcup_{t \in Y_t} \Pot(y))\leqslant 3(r^*_{FkC}(P)+2 r_Q).$$

       $S$ is a feasible solution covering all points in $\bigcup_t \bigcup_{t \in Y_t} \Pot(y)$. $\forall x\in P$, suppose that it is assigned to processor $t$. Then, by definition of $Y_t$, there exists $y\in Y_t$  such that $d(x,y)\leqslant r^Q_t$.
       Thus, $d(x,S)\leqslant 3(r^*_{FkC}(P)+3 r_Q)$.

        Similarly, with $Q$ increasing, $r_Q$ is  decreasing. Once $r_Q\leqslant \frac{\epsilon}{3}r^*_{FkC}(P)$, we have $$d(x,S)\leqslant 3(1+\epsilon)r^*_{FkC}(P)$$

        The theorem is proved.
    \end{proof}

    {\bf Proof of \cref{lem:coreset_of_Streaming}}
        \begin{proof}
      
         $r(t)$ is maintained by \cite{charikar2004incremental}'s incremental algorithm. Hence,  $r(t)\leqslant r^*(t)$, i.e., $\frac{\Bar{\epsilon}}{2} r(t)\leqslant \frac{\Bar{\epsilon}}{2}r^*$.

          {We now prove,  by induction, that  $Y(t)$ as constructed is an $(\frac{\Bar{\epsilon} r(t)}{2},2)$ net of first $t$ points. %The proof is by induction. 
        Initially, when $t\leqslant k+1,$ $Y(t-1)$ contains all points read so far and it is obviously an $(\frac{\Bar{\epsilon} r(t-1)}{2},2)$ net $(r(t-1)=0)$. }
        
        {Suppose that after processing  $p_{t-1},$ $Y(t-1)$ is an $(\frac{\Bar{\epsilon} r(t-1)}{2},2)$ net of $P_{t-1}.$ $p_t$ is then read.
        If  \cite{charikar2004incremental} calculates that  $r(t)=2^\lambda r(t-1)$ for some $\lambda>0$,  \cref{ALg:Streaming} runs line 8 and computes a new net. }
        
        {By \cref{lem:update_r}, the new $Y(t)-1$ is a $(\frac{\Bar{\epsilon} r(t)}{2},2)$ net of $P_{t-1}.$  Then, after  running lines 10-19 in \cref{ALg:Streaming}, $Y(t)$ is a $(\frac{\Bar{\epsilon} r(t)}{2},2)$ net of $P_t.$. }

        Thus, after step $t,$, $Y(t)$ is an $(\frac{\Bar{\epsilon}}{2} r(t),2)$-net of $P_t$.
        
        When $p_t$ is read, if there exists a point $y\in Y(t)$ such that $d(p_t,y)\leqslant \Bar{\epsilon} r_t$, we can add $p$ into $N(y,\frac{\Bar{\epsilon}}{2} r(t))$ and update $\COL_g(p)(y)$ $\Pot(y,g(p))$. Otherwise, $p_t$ is added into $Y(t)$ with $N(p_t)=\{p_t\}$. 
        
         Since (1) $\Bar{\epsilon} r_t\leqslant \frac{\Bar{\epsilon}}{2}r^*$ (2) $\bigcup_{y\in Y_t} N(y)=P_t$ (3) $\forall x,y\in Y_t$ $N(x)\cap N(y)=\emptyset$, $Y_t$ is proper.

        From property 1 of \cite{charikar2004incremental}'s algorithm,  $P_t$ can be covered by $k$ balls $B(s,8 r(t))$ satisfying % with radius $8r(t)$; in each ball, we have 
        $$\forall s \in S(t),\quad
        \max_{x,y\in B(s,8(r(t))}d(x,y)\leqslant 16 r(t).$$
        
        In addition, $\forall x,y\in Y_t,$ $d(x,y)\geqslant \frac{\Bar{\epsilon}}{2} r(t)$. The aspect ratio of each ball is then at most $\frac{2\cdot 16r(t)}{\Bar{\epsilon} r(t)}\leqslant \frac{32}{\Bar{\epsilon}}$. By \cref{lem:doubling}, the number of points in each ball is at most $O((\frac{32}{\Bar{\epsilon}})^{O(D)})$. Therefore, $|Y(n)|=O(k(32/\Bar{\epsilon})^D)$.
    \end{proof}

    {\bf Proof of \cref{thm:heuristic_streaming}}
    \begin{proof}
    By \cite{charikar2004incremental}, $r(t)$ maintained is the lower bound of $r^*_{QC}(P(t))$, i.e., the optimal value of the $Q$ center problem on point set $P(t)$. With $Q$  is increasing, both $r(t)$ and $r^*_{QC}(P(t))$ are decreasing. 
    
    Once $r(t)\leqslant \frac{\epsilon}{24}r^*_{FkC}(P(t))$ $\forall t$, we can prove, by induction, that  $Y_{r(t)}$ is an $\epsilon$-proper coreset of $P(t)$. 

    Initially, when $t\leqslant Q+1,$ $Y(t-1)$ contains all points read so far and it is obviously an $(r(t-1),2)$ net $(r(t-1)=0)$. 
    
    Suppose that after processing  $p_{t-1},$ $Y(t-1)$ is an $( r(t-1),2)$ net of $P_{t-1}.$ $p_t$ is then read.
    If  \cite{charikar2004incremental} calculates that  $r(t)=2^\lambda r(t-1)$ for some $\lambda>0$,  our heuristic apply \cref{algo:coreset} computes a new net. 
    
    By \cref{lem:update_r}, the new $Y(t)-1$ is a $( r(t),2)$ net of $P_{t-1}.$  Then, insert $t+1$th point, $Y(t)$ is a $( r(t),2)$ net of $P_t.$. 

    Thus, after step $t,$, $Y(t)$ is an $(r(t),2)$-net of $P_t$ and $Y(t)$ is $\epsilon$ proper.

    Hence, when $Q$ is large enough, by \cref{lem:approximation_in_coreset} this heuristic algorithm is a $3(1+\epsilon)$ approximation algorithm.
    \end{proof}

\section{Heuristic  streaming algorithm}

Initially, when we read the first $t\leqslant Q$ points, we record $Y_{r(t)}=P(t)$. $\forall x\in Y_{r(t)}$ from group $i$, we define $\COL_i(x)=1$ and $\Pot(x,i)=x$. 

When $t=Q+1$, we compute $r(t)=\frac{\min_{x,y\in P(t)}d(x,y)}{2}$. Then we compute a new $Y_{r(t)}\subset P(t)$ such that $\forall x,y\in Y_{r(t)}, d(x,y)>4r(t)$ and $\max_{x\in P(t)}d(x,Y_{r(t)})\leqslant 8r(t)$, i.e., $Y_{r(t)}$ is a $(4r(t),2)$-net of $P(t)$. In addition, $\forall x\in P(t)$ from group $i$, we assign $x$ to its closest point $y\in Y_{r(t)}$ and set $\COL_i(y)=1$. If $\Pot(y,i)$ is defined and $d(\Pot(y,i),y)>d(x,y)$, we update $\Pot(y,i)=x$. If $\Pot(y,i)$ is not defined, we set $\Pot(y,i)=x$.

    For $t > Q+1$, Similar to \cite{charikar2004incremental}, when we read  $p_{t+1}$ and find that it is in  group $i$, we consider the following cases:
    \begin{itemize}
        \item if there exists $y\in Y_{r(t)}$ such that $d(p_{t+1},y)\leqslant 8r(t)$,  $Y_{r(t+1)}=Y_{r(t)}$ and $r(t+1)=r(t)$. Set $\COL_i(y)=1$. \\If $\Pot(y,i)$ is defined and $d(\Pot(y,i),y)>d(x,y)$, we update $\Pot(y,i)=x$. \\If $\Pot(y,i)$ is not defined, we set $\Pot(y,i)=x$.
        \item if $\forall y\in Y_{r(t)}$ satisfying $d(p_{t+1},y)> 8r(t)$ and $|Y_{r(t)}|<Q$, $Y_{r(t+1)}=Y_{r(t)}\cup \{p_{t+1}\}$ while $r(t+1)=r(t)$.\\ Set $\COL_i(p_{t+1})=1$ and $\Pot(p_{t+1},i)=p_{t+1}$.
        \item if $\forall y\in Y_{r(t)}$ satisfying $d(p_{t+1},y)> 8r(t)$ and $|Y_{r(t)}|=Q$, we need to update the lower bound $r(t)$. \\Let $Y'(\lambda)$ be the maximal subset of  $Y_{r(t)}\cup \{p_{t+1}\}$ such that $\forall y_1,y_2\in Y'(\lambda)$ $d(y_1,y_2)> 4\cdot 2^\lambda r(t)$. \\Note  that $|Y'(0)|=|Y_{r(t)}S(t)\cup \{p_{t+1}\}|=k+1$. \\
        Now compute the smallest integer $\lambda$ such that $|Y'(\lambda)|\leqslant Q$. Then,  set $Y_{r(t+1)}=Y'(\lambda)$ and $r(t+1)=2^\lambda r(t)$.\\ $\forall y\in Y_{r(t)}/Y_{r(t+1)}$, select a $y'\in Y_{r(t+1)}$. \\Define $\COL_i(y')=\COL_i(y')\vee \COL_i(y)$. \\If $\Pot(y',i)$ is defined and $d(\Pot(y',i),y')>d(\Pot(y,i),y')$, we update $\Pot(y',i)=\Pot(y,i)$. \\If $\Pot(y,i)$ is not defined, we set $\Pot(y',i)=\Pot(y,i)$. 
    \end{itemize}

\section{Experiments on small size datasets}
\paragraph{Small Size Datasets}
{Our goal was to design algorithms for larger datasets. Our experiments were run on such datasets   and demonstrated the advantages of our algorithms when run on them. An obvious followup question would be  how they would perform on smaller datasets.}

%\mjgtext{
To address this we replicate the same dataset parameters as in \cite{chiplunkar2020solve} and show the result in Table \ref{tab:small}. Indeed, our  algorithms will sometimes not  work as well as the others for small data sets.

In the streaming setting, our heuristic algorithm essentially runs  \cite{charikar2004incremental}'s incremental algorithm   for the $Q$-center problem to construct a coreset with $Q$ points. When the incremental algorithm doubles the lower bound $r(t)$, we update the current coreset with $Q$ points to a new coreset with $Q'$ points. In a  small size data set, when $r(t)$ gets large, $Q'$ could easily become very small and very small coresets can more easily yield bad approximations.
%and be hard to extend. 
%Finally, our heuristic algorithm would construct a coreset with a small size $Q'<<Q$. The solution generated from $Q'$ is not ideal.

For the MapReduce setting, the 
%message conveyed is
results are generally the same as for large datasets; our algorithms will have appropriate solution ratios in most cases.
%}

%\mjgtext{PLEASE FIX THIS. I REALLY DID NOT UNDERSTAND THE SENTENCE
%''For example, querying One Pass algorithm just after doubling on small datasets will possibly almost reach the upper bound. ''}

%Because we select much more points to show the advantage of our algorithms than the previous analysis in \cite{chiplunkar2020solve}, the final question may be: will they perform bad when data is not such large? For this reason, we replicate the same dataset settings as Chiplunkar et al., and list the result in Table \ref{tab:small}. Indeed, the algorithms may fall in to bad cases more easily when dataset is not large: for example, querying One Pass algorithm just after doubling on small datasets will possibly almost reach the upper bound. That said, the message conveyed is generally the same as large datasets, our algorithms will have appropriate solution ratios in most cases.

\begin{table}[!ht]
    \centering 
    \begin{tabular}{|c|c|c?c|c|c?c|c|}
    \hline
        Dataset & Capacities & \makecell{Lower\\Bound} & JNN & Two Pass & One Pass & \makecell{CKR\\Distributed}  & \makecell{Map\\Reduce} \\ \hline
        \multirow{2}{*}{CelebA}  &  [2, 2]  &  30142  & \cellcolor{black!1} 1.95 & \cellcolor{black!8} 1.76 & \cellcolor{black!2} 1.88 & \cellcolor{black!8} 1.76 & \cellcolor{black!25} 1.67 \\ \cline{2-8}
 &   [2] * 4  &  28247  & \cellcolor{black!2} 1.93 & \cellcolor{black!4} 1.88 & \cellcolor{black!1} 2.02 & \cellcolor{black!4} 1.88 & \cellcolor{black!25} 1.72 \\ \cline{1-1} \cline{2-8}
\multirow{3}{*}{SushiA}  &   [2, 2]  &  11.00  & \cellcolor{black!25} 2.00 & \cellcolor{black!25} 2.00 & \cellcolor{black!4} 2.18 & \cellcolor{black!10} 2.09 & \cellcolor{black!1} 2.27 \\ \cline{2-8}
 &   [2] * 6  &  8.50  & \cellcolor{black!25} 2.12 & \cellcolor{black!3} 2.35 & \cellcolor{black!3} 2.35 & \cellcolor{black!8} 2.24 & \cellcolor{black!8} 2.24 \\ \cline{2-8}
 &   [2] * 12  &  7.50  & \cellcolor{black!6} 2.27 & \cellcolor{black!2} 2.40 & \cellcolor{black!6} 2.27 & \cellcolor{black!2} 2.40 & \cellcolor{black!25} 2.13 \\ \cline{1-1}\cline{2-8}
 &   [2, 2]  &  35.00  & \cellcolor{black!2} 2.03 & \cellcolor{black!25} 1.80 & \cellcolor{black!25} 1.80 & \cellcolor{black!12} 1.86 & \cellcolor{black!3} 2.00 \\ \cline{2-8}
SushiB\tablefootnote{We've noticed that the SushiB result is different from Chiplunkar et al.'s results. We run their code exactly as is and produced the same result in this table. There might be issues about dataset consistencies.}  &   [2] * 6  &  32.50  & \cellcolor{black!6} 1.94 & \cellcolor{black!25} 1.82 & \cellcolor{black!13} 1.88 & \cellcolor{black!13} 1.88 & \cellcolor{black!3} 2.00 \\ \cline{2-8}
 &   [2] * 12  &  30.50  & \cellcolor{black!12} 2.00 & \cellcolor{black!12} 2.00 & \cellcolor{black!25} 1.93 & \cellcolor{black!12} 2.00 & \cellcolor{black!25} 1.93 \\ \cline{1-1}\cline{2-8}
 &   [2, 2]  &  4.90  & \cellcolor{black!0} 2.34 & \cellcolor{black!25} 1.90 & \cellcolor{black!0} 2.44 & \cellcolor{black!7} 2.02 & \cellcolor{black!0} 2.34 \\ \cline{2-8}
Adult  &   [2] * 5  &  3.92  & \cellcolor{black!0} 2.48 & \cellcolor{black!0} 2.36 & \cellcolor{black!25} 1.93 & \cellcolor{black!0} 2.35 & \cellcolor{black!1} 2.25 \\ \cline{2-8}
 &   [2] * 10  &  2.76  & \cellcolor{black!6} 2.64 & \cellcolor{black!25} 2.47 & \cellcolor{black!0} 2.95 & \cellcolor{black!2} 2.75 & \cellcolor{black!0} 2.92 \\ \hline
    \end{tabular}
    \caption{\label{tab:small}Costs on first $1\,000$ points of datasets. Number of cores is set to $40$ for distributed algorithms.}
\end{table}

%% file: example_paper.bbl
\begin{thebibliography}{24}
\providecommand{\natexlab}[1]{#1}
\providecommand{\url}[1]{\texttt{#1}}
\expandafter\ifx\csname urlstyle\endcsname\relax
  \providecommand{\doi}[1]{doi: #1}\else
  \providecommand{\doi}{doi: \begingroup \urlstyle{rm}\Url}\fi

\bibitem[Angelidakis et~al.(2022)Angelidakis, Kurpisz, Sering, and
  Zenklusen]{angelidakis2022fair}
Angelidakis, H., Kurpisz, A., Sering, L., and Zenklusen, R.
\newblock Fair and fast k-center clustering for data summarization.
\newblock In \emph{International Conference on Machine Learning}, pp.\
  669--702. PMLR, 2022.

\bibitem[Bera et~al.(2019)Bera, Chakrabarty, Flores, and
  Negahbani]{bera2019fair}
Bera, S., Chakrabarty, D., Flores, N., and Negahbani, M.
\newblock Fair algorithms for clustering.
\newblock \emph{Advances in Neural Information Processing Systems}, 32, 2019.

\bibitem[Bera et~al.(2022)Bera, Das, Galhotra, and Kale]{bera2022fair}
Bera, S.~K., Das, S., Galhotra, S., and Kale, S.~S.
\newblock Fair k-center clustering in mapreduce and streaming settings.
\newblock In \emph{Proceedings of the ACM Web Conference 2022}, pp.\
  1414--1422, 2022.

\bibitem[Bercea et~al.(2019)Bercea, Gro{\ss}, Khuller, Kumar, R{\"o}sner,
  Schmidt, and Schmidt]{bercea2019cost}
Bercea, I.~O., Gro{\ss}, M., Khuller, S., Kumar, A., R{\"o}sner, C., Schmidt,
  D.~R., and Schmidt, M.
\newblock On the cost of essentially fair clusterings.
\newblock In \emph{Approximation, Randomization, and Combinatorial
  Optimization. Algorithms and Techniques (APPROX/RANDOM 2019)}. Schloss
  Dagstuhl-Leibniz-Zentrum fuer Informatik, 2019.

\bibitem[Ceccarello et~al.(2019)Ceccarello, Pietracaprina, and
  Pucci]{ceccarello2019solving}
Ceccarello, M., Pietracaprina, A., and Pucci, G.
\newblock Solving k-center clustering (with outliers) in mapreduce and
  streaming, almost as accurately as sequentially.
\newblock \emph{Proceedings of the VLDB Endowment}, 12\penalty0 (7):\penalty0
  766--778, 2019.

\bibitem[Charikar et~al.(2004)Charikar, Chekuri, Feder, and
  Motwani]{charikar2004incremental}
Charikar, M., Chekuri, C., Feder, T., and Motwani, R.
\newblock Incremental clustering and dynamic information retrieval.
\newblock \emph{SIAM Journal on Computing}, 33\penalty0 (6):\penalty0
  1417--1440, 2004.

\bibitem[Chen et~al.(2016)Chen, Li, Liang, and Wang]{chen2016matroid}
Chen, D.~Z., Li, J., Liang, H., and Wang, H.
\newblock Matroid and knapsack center problems.
\newblock \emph{Algorithmica}, 75\penalty0 (1):\penalty0 27--52, 2016.

\bibitem[Chierichetti et~al.(2017)Chierichetti, Kumar, Lattanzi, and
  Vassilvitskii]{chierichetti2017fair}
Chierichetti, F., Kumar, R., Lattanzi, S., and Vassilvitskii, S.
\newblock Fair clustering through fairlets.
\newblock \emph{Advances in Neural Information Processing Systems}, 30, 2017.

\bibitem[Chiplunkar et~al.(2020)Chiplunkar, Kale, and
  Ramamoorthy]{chiplunkar2020solve}
Chiplunkar, A., Kale, S., and Ramamoorthy, S.~N.
\newblock How to solve fair $k$-center in massive data models.
\newblock In \emph{International Conference on Machine Learning}, pp.\
  1877--1886. PMLR, 2020.

\bibitem[Dean \& Ghemawat(2008)Dean and Ghemawat]{dean2008mapreduce}
Dean, J. and Ghemawat, S.
\newblock Mapreduce: simplified data processing on large clusters.
\newblock \emph{Communications of the ACM}, 51\penalty0 (1):\penalty0 107--113,
  2008.

\bibitem[Ding et~al.(2023)Ding, Huang, Liu, Yu, and Wang]{ding2023randomized}
Ding, H., Huang, R., Liu, K., Yu, H., and Wang, Z.
\newblock Randomized greedy algorithms and composable coreset for k-center
  clustering with outliers.
\newblock \emph{arXiv preprint arXiv:2301.02814}, 2023.

\bibitem[Gonzalez(1985)]{gonzalez1985clustering}
Gonzalez, T.~F.
\newblock Clustering to minimize the maximum intercluster distance.
\newblock \emph{Theoretical Computer Science}, 38:\penalty0 293--306, 1985.

\bibitem[Heinonen et~al.(2001)]{heinonen2001lectures}
Heinonen, J. et~al.
\newblock \emph{Lectures on analysis on metric spaces}.
\newblock Springer Science \& Business Media, 2001.

\bibitem[Hochbaum \& Shmoys(1985)Hochbaum and Shmoys]{hochbaum1985best}
Hochbaum, D.~S. and Shmoys, D.~B.
\newblock A best possible heuristic for the k-center problem.
\newblock \emph{Mathematics of operations research}, 10\penalty0 (2):\penalty0
  180--184, 1985.

\bibitem[Hsu \& Nemhauser(1979)Hsu and Nemhauser]{hsu1979easy}
Hsu, W.-L. and Nemhauser, G.~L.
\newblock Easy and hard bottleneck location problems.
\newblock \emph{Discrete Applied Mathematics}, 1\penalty0 (3):\penalty0
  209--215, 1979.

\bibitem[Jones et~al.(2020)Jones, Nguyen, and Nguyen]{jones2020fair}
Jones, M., Nguyen, H., and Nguyen, T.
\newblock Fair k-centers via maximum matching.
\newblock In \emph{International Conference on Machine Learning}, pp.\
  4940--4949. PMLR, 2020.

\bibitem[Kamishima()]{SushitData}
Kamishima, T.
\newblock Sushi preference data sets.
\newblock URL \url{https://www.kamishima.net/sushi/}.

\bibitem[Kay et~al.(2015)Kay, Matuszek, and Munson]{kay2015unequal}
Kay, M., Matuszek, C., and Munson, S.~A.
\newblock Unequal representation and gender stereotypes in image search results
  for occupations.
\newblock In \emph{Proceedings of the 33rd annual acm conference on human
  factors in computing systems}, pp.\  3819--3828, 2015.

\bibitem[Kleindessner et~al.(2019)Kleindessner, Awasthi, and
  Morgenstern]{kleindessner2019fair}
Kleindessner, M., Awasthi, P., and Morgenstern, J.
\newblock Fair k-center clustering for data summarization.
\newblock In \emph{International Conference on Machine Learning}, pp.\
  3448--3457. PMLR, 2019.

\bibitem[Kohavi \& Becker()Kohavi and Becker]{AdultData}
Kohavi, R. and Becker, B.
\newblock Adult data set.
\newblock URL \url{https://archive.ics.uci.edu/ml/datasets/Adult}.

\bibitem[Krauthgamer \& Lee(2004)Krauthgamer and
  Lee]{krauthgamer2004navigating}
Krauthgamer, R. and Lee, J.~R.
\newblock Navigating nets: Simple algorithms for proximity search.
\newblock In \emph{Proceedings of the fifteenth annual ACM-SIAM symposium on
  Discrete algorithms (SODA)}, pp.\  798--807, 2004.

\bibitem[Liu et~al.(2015)Liu, Luo, Wang, and Tang]{liu2015deep}
Liu, Z., Luo, P., Wang, X., and Tang, X.
\newblock Deep learning face attributes in the wild.
\newblock In \emph{Proceedings of the IEEE international conference on computer
  vision}, pp.\  3730--3738, 2015.

\bibitem[Talwar(2004)]{talwar2004bypassing}
Talwar, K.
\newblock Bypassing the embedding: algorithms for low dimensional metrics.
\newblock In \emph{Proceedings of the thirty-sixth annual ACM symposium on
  Theory of computing}, pp.\  281--290, 2004.

\bibitem[Yuan et~al.(2021)Yuan, Diao, Du, and Liu]{yuan2021distributed}
Yuan, F., Diao, L., Du, D., and Liu, L.
\newblock Distributed fair k-center clustering problems with outliers.
\newblock In \emph{International Conference on Parallel and Distributed
  Computing: Applications and Technologies}, pp.\  430--440. Springer, 2021.

\end{thebibliography}
